\documentclass[sigconf,nonacm,screen]{acmart}
\AtBeginDocument{%
  \providecommand\BibTeX{{%
    \normalfont B\kern-0.5em{\scshape i\kern-0.25em b}\kern-0.8em\TeX}}}

\newcommand\vldbpagestyle{plain}
\usepackage{style}
  \newcommand{\revision}[1]{#1}

\definecolor{forestgreen}{rgb}{0.13, 0.55, 0.13}

\newcommand{\modelop}[1]{\texttt{#1}}
\newcommand{\forkins}{\modelop{fork}}

\newcommand{\thread}{thread}

\newcommand{\ifconference}{{{\ifx\fullversion\undefined}}}

\makeatletter
\def\dfnt@space@setup{%
\dfnt@preskip=\parskip
  \dfnt@postskip=0pt}
\makeatother

\newtheoremstyle{exampstyle}
{.05in} % Space above
{.05in} % Space below
{} % Body font
{.5em} % Indent amount
{\sc \bfseries} % Theorem head font
{.} % Punctuation after theorem head
{.5em} % Space after theorem head
{} % Theorem head spec (can be left empty, meaning `normal')
\theoremstyle{exampstyle} 
\theoremstyle{exampstyle} 

\makeatletter
\renewenvironment{proof}[1][\proofname]{\par
\vspace{-\topsep}% remove the space after the theorem
\pushQED{\qed}%
\normalfont
\topsep0pt \partopsep0pt % no space before
\trivlist
\item[\hskip\labelsep
      \itshape
  #1\@addpunct{.}]\ignorespaces
}{%
\popQED\endtrivlist\@endpefalse
\addvspace{3pt plus 3pt} % some space after
}

\newcommand{\cf}{\ensuremath{\mathcal{C}}\xspace}
\newcommand{\lmax}{\ensuremath{L_{\max}}}
\newcommand{\level}[1]{\ensuremath{l(#1)}}
\newcommand{\cluster}[1]{\ensuremath{\mathsf{cluster}(#1)}}
\newcommand{\sizeconst}[1]{\ensuremath{2^{ #1 }}}

 \crefname{section}{Sec.}{Sec.}
 \crefname{theorem}{Thm.}{Thm.}
 \crefname{lemma}{Lem.}{Lem.}
 \crefname{corollary}{Col.}{Col.}
 \crefname{table}{Tab.}{Tabs.}
 \crefname{algorithm}{Alg.}{Algs.}
 \Crefname{table}{Tab.}{Tabs.}

\setcopyright{acmlicensed}
\copyrightyear{2018}
\acmYear{2018}
\acmDOI{XXXXXXX.XXXXXXX}

\acmConference[Conference acronym 'XX]{Make sure to enter the correct
  conference title from your rights confirmation emai}{June 03--05,
  2018}{Woodstock, NY}
\acmISBN{978-1-4503-XXXX-X/18/06}

\begin{document}


%\title{Linear Space Parallel Batch-Dynamic Graph Connectivity}
\title{Towards Scalable and Practical Batch-Dynamic Connectivity}


%\hide{

\settopmatter{authorsperrow=3}

\author{Quinten De Man}
\affiliation{
    \institution{University of Maryland}
    \country{}
}
\email{deman@umd.edu}

\author{Laxman Dhulipala}
\affiliation{
    \institution{University of Maryland}
    \country{}
}
\email{laxman@umd.edu}

\author{Adam Karczmarz}
\affiliation{
    \institution{University of Warsaw \& IDEAS NCBR}
    \country{}
}
\email{a.karczmarz@mimuw.edu.pl}

\author{Jakub Łącki}
\affiliation{%
  \institution{Google Research}
    \country{}
}
\email{jlacki@google.com}

\author{Julian Shun}
\affiliation{
    \institution{MIT CSAIL}
    \country{}
}
\email{jshun@mit.edu}

\author{Zhongqi Wang}
\affiliation{
    \institution{University of Maryland}
    \country{}
}
\email{zqwang@umd.edu}

%}


\begin{abstract}
We study the problem of dynamically maintaining the connected components of an undirected graph subject to edge insertions and deletions. We give the first parallel algorithm for the problem which is work-efficient, supports batches of updates, runs in polylogarithmic depth, and uses only linear total space. The existing algorithms for the problem either use super-linear space, do not come with strong theoretical bounds, or are not parallel.

On the empirical side, we provide the first implementation of the \emph{cluster forest algorithm}, the first linear-space and poly-logarithmic update time algorithm for dynamic connectivity. 
\revision{Experimentally, we find that our algorithm uses up to $19.7\times$ less space and is up to $6.2\times$ faster than the level-set algorithm of HDT, arguably the most widely-implemented dynamic connectivity algorithm with strong theoretical guarantees.}

\end{abstract}


\hide{


}

\renewcommand\footnotetextcopyrightpermission[1]{} % This line removes the footnote about the conference and year.



\maketitle

\pagestyle{\vldbpagestyle}

\section{Introduction}
The problem of dynamically maintaining the connected components of a
graph is a critical subroutine that is often used when designing dynamic
algorithms for other fundamental and practical problems, e.g., dynamic
DBSCAN~\cite{khan2014dbscan,GT2017}, hierarchical clustering~\cite{monath2023online,Ruan2021}, approximate MST~\cite{connectit}, among other problems.
%\laxman{(todo:)$\ldots$ add cites/mot from connectit}.\kuba{Some ideas: flat single-linkage clustering, approximate MST; there is also unique perfect matching, but that's probably too exotic}
%
It is also one of the most intensely studied dynamic graph problems,
and has seen extensive algorithmic development over the past three
decades~\cite{henzinger1995randomized, eppstein1997sparsification,
holm2001poly, thorup1999decremental, thorup2000near,
henzinger2001maintaining, wulff2013faster, kapron2013dynamic,
rasumssen16faster, huang17dynconn, nanongkai2017dynamic,
wulff2017fully}.
In the fully-dynamic graph connectivity problem, the goal is to build
a data structure that supports the following three operations on a
dynamic undirected graph $G=(V,E)$ with a fixed set of vertices:
\begin{itemize}
    \item $\textbf{Insert}(u,v)$ inserts edge $(u,v)$ into $G$.
    \item $\textbf{Delete}(u,v)$ deletes edge $(u,v)$ from $G$.
    \item $\textbf{Connected}(u,v)$ determines whether vertices $u$
    and $v$ are in the same connected component of $G$.
\end{itemize}

Despite its importance, dynamic connectivity has not yet been
efficiently solved in practice, and marks a major gap in our
understanding of how to bridge theory and practice in dynamic
algorithms.
In the special case of dynamic forest-connectivity, there are data structures such as
{\em Euler Tour Trees (ETTs)}~\cite{henzinger1995randomized} that are reasonably practical and have
been adapted to support parallel updates~\cite{tseng2019batch}.
However,  few implementations exist for the more complex case of general graphs.
For general graphs, perhaps the best known dynamic connectivity data
structure that provides good theoretical guarantees and
has been implemented is due to Holm, de Lichtenberg, and
Thorup~\cite{holm2001poly}, who developed a dynamic connectivity algorithm that performs updates in $O(\log^2 n)$ amortized time and requires $O(n\log n + m)$ space, for a graph with $n$ vertices and $m$ edges.
Their algorithm is based on $O(\log n)$ layers of dynamic forest-connectivity data structures; we refer to their idea as the {\em HDT algorithm}.

However, existing implementations of the HDT algorithm suffer from
high overheads in space and time, limiting their practical
applicability.
For example, we found that to run on a 1.2 billion edge graph, an optimized dynamic connectivity implementation based on the HDT algorithm requires
up to 360 billion bytes---over $70$ times more than what it takes to store the graph using a simple (static) representation.
%\kuba{I removed the comparison to  union-find since (a) we don't compare to union find empirically and don't want to steer reviewers attention there (b) it's a bit unfair as union find does not need to store the edges}\laxman{sounds good!}
%
% Similarly, another recent work optimizing the HDT algorithm, D-tree,\kuba{Add citation}
% despite spending significant effort optimizing the HDT algorithm
% requires C bytes per edge [add any other performance issues].
%\quinten{I removed mention of D-tree here since it is not HDT based.}
%
The key limitation of implementations based on the HDT algorithm is
that the connectivity information is stored redundantly in a logarithmic number of 
%\kuba{can we be more specific and say that it's a logarithmic number? This may be stronger than just 'multiple'.}
forest-connectivity data structures across different layers.
At a high level, the HDT algorithm maintains a spanning forest $F$ and hierarchy of nested edge 
subsets $F_1 \subseteq F_2 \subseteq \ldots F_k = F$. With this representation, \revision{the vertices of $G$ are present in $k$ trees}, which results in a space usage of 
$O(n\log n)$ across all of the trees.
% \quinten{This confuses me. Shouldn't $F_0$ have no edges? Do you mean the vertices are present in all trees?}\laxman{That's true. I updated it to say $F_1$---does that make sense?}

On the theoretical side, the space usage was improved to linear by an elegant {\em cluster forest algorithm (CF algorithm)} that is inspired by the HDT algorithm~\cite{thorup2000near}.
The key idea in the CF algorithm is to store a single forest of trees (called
the {\em cluster forest}) that implicitly represents the
connectivity information stored in the nested $O(\log n)$ layers of
the HDT algorithm.
%\kuba{At this point, it's a bit unclear what we mean by laminar. Would it make sense to say something like:}
%
In its basic version, the CF algorithm achieves similar update times to HDT, while improving the space bound to $O(n+m)$. 
This was the first algorithm that solved dynamic graph connectivity in
poly-logarithmic update time and linear space. 
The CF algorithm was later simplified and optimized by Wulff-Nilsen~\cite{wulff2013faster}.
%\julian{we should mention whether the HDT and Thorup are deterministic?}\laxman{I think we should for HDT. Or maybe we can mention these properties more carefully in the related work?}
%\julian{sounds good}

\begin{figure*}[ht]
    \vspace{-2em}
    \includegraphics[width=0.98\textwidth]{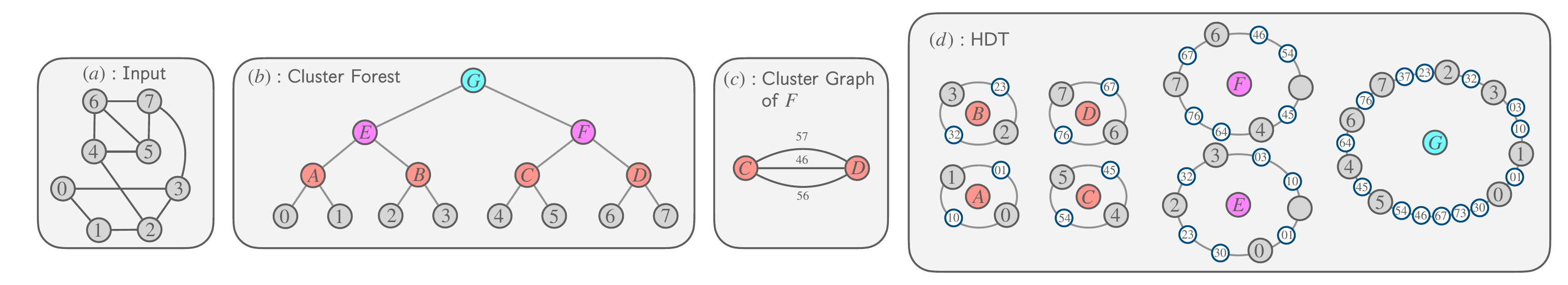}
    \caption{\small {\bf The core data structures used by the cluster forest (CF) and HDT algorithms.} The input graph is shown in \textbf{(a)}. 
    The cluster forest is given in \textbf{(b)}, and represents the nested hierarchy of connected components. 
    The leaves (level 0 nodes) are the original vertices. 
    The internal nodes (level $> 0$) are colored using different colors per-level and are the nested components. \textbf{(c)} shows the cluster graph of the level 2 component $F$. 
    The nodes in the cluster graph are the children of $F$, and the edges in the cluster graph are the level $2$ edges that are incident to vertices that $F$ contains (the level $2$ edges go between level $1$ components). 
    Lastly, \textbf{(d)} shows the same component hierarchy as (b), but as stored by the HDT algorithm. The HDT algorithm stores a {\em separate} Euler Tour Tree (ETT) for {\em every} component in the hierarchy. Each tree edge is stored twice (illustrated as smaller white circles), once per direction.}
%    \julian{it would be helpful to specify what the color of the vertices means}}
    \vspace{-2em}
    \label{fig:introfig}
\end{figure*}

Given the theoretical advantages of the CF algorithm, an important question
is: {\bf \emph{is the CF algorithm practical?}}
For example, does it yield improved space-efficiency or
faster runtime in practice relative to existing dynamic connectivity
implementations?
This question is highly non-trivial due to the complexity of implementing and
optimizing the CF algorithm, which uses a significant amount of indirection
and requires performing sophisticated tree traversals and amortization to 
obtain its update bounds.
A second important question is: {\bf \emph{can the CF algorithm be efficiently parallelized?}}
%\julian{motivate the batch-dynamic setting}
In particular, can we ensure that each update is processed with low \emph{depth} (longest chain of sequentially
dependent instructions) even in the worst-case?
Furthermore, can we also make the algorithm work efficiently in the \emph{parallel batch-dynamic setting}?
We note that the batch-dynamic setting, in which updates come in batches of arbitrary size, is the  standard modern setting for parallelizing dynamic algorithms~\cite{acar2019parallel, hanauer2022recent, pandey2021terrace, liu2022parallel, besta2019practice, gabert2021shared, dhulipala2020parallel}.
Ideally we would like to parallelize the algorithm without sacrificing
space-efficiency or work-efficiency.
That is, each batch of updates should be performed with low depth, and work (total number of operations) and space matching that of the sequential CF algorithm.
We note that while the HDT algorithm was recently shown to be amenable to an
efficient batch-dynamic algorithm~\cite{acar2019parallel}, it is not space-efficient.

In this paper, we carefully study the CF approach in theory and practice
to answer these open questions.
On the theoretical side, we extend the CF algorithm and show how to achieve low depth.
Specifically, we introduce a new invariant (the
{\em blocked invariant}),
which provides important additional structure that we exploit.
Using our new invariant and our approach to maintaining it, we obtain the first space-efficient and work-efficient parallel algorithm that has poly-logarithmic depth.

On the empirical side, we perform the first experimental study of the CF approach in the sequential setting.
\revision{Compared with the existing state-of-the-art dynamic connectivity implementations with worst-case guarantees based on the HDT algorithm, we find that our implementations use up to $19.7\times$ less space and are up to $6.2\times$ faster than an optimized implementation of HDT.}
%Compared with the existing state-of-the-art dynamic connectivity implementations
%based on the HDT algorithm, our new CF implementations are faster and more space efficient (up to 22$\times$ less peak space).
%We also find that the CF approach is up to 620$\times$ faster than the D-tree~\cite{chen2022spanning}, \quinten{update numbers here} a recent linear-space dynamic  connectivity algorithm that maintains a single spanning forest on the graph. 
%
%Their approach sacrifices worst-case theoretical guarantees
%and attempts to heuristically optimize the average performance on certain datasets.
%
In the next two sections, we formalize  the data structures, and present a technical overview of our results.

\section{Preliminaries}

% \laxman{The following could be pushed to the appendix and only referenced if used}
% \paragraph{Useful Lemmas}
% The following lemmas proved in \cite{acar2019parallel} are useful for analyzing the work bounds of our
% parallel algorithms.
% \begin{lemma}
% \label{lem:component_bounds}
%     Let $n_1, n_2, \ldots, n_c$ and $k_1, k_2, \ldots, k_c$ be sequences of non-negative integers such that $\sum k_i = k$ and $\sum n_i = n$. Then
%     $$\sum_{i=1}^c k_i \log\left(1 + \frac{n_i}{k_i}\right) \leq k \log\left(1 + \frac{n}{k}\right).$$
% \end{lemma}

% \begin{lemma}
% \label{lem:batch_bound_is_root_dominated}
%     For any non-negative integers $n$ and $r$,
%     $$\sum_{w=0}^{r} 2^{w} \log \left( 1 + \frac{n}{2^{w}} \right) = O\left( 2^{r}\log\left(1 + \frac{n}{2^{r}} \right) \right).$$
% \end{lemma}

% \begin{lemma}
% \label{lem:batch_bound_is_increasing}
%     For any $n \geq 1$, the function $x \log(1+n/x)$ is strictly increasing with respect to $x$ for $x \geq 1$.
% \end{lemma}

\myparagraph{Model}
We use the work-span (or work-depth) model for fork-join parallelism to analyze parallel algorithms~\cite{CLRS,blelloch2020optimal}.
The model assumes a set of \thread{}s that share memory.
A \thread{} can \forkins{} $k$ child \thread{s} that run in parallel.
When all children complete, the parent \thread{} continues.
The \defn{work} \revision{$W$} of an algorithm is the total number of instructions and
the \defn{span} (depth) \revision{$D$} is the length of the longest sequence of dependent instructions.
Computations can be executed using a randomized work-stealing
scheduler in practice in $W/P+O(D)$ time \whp{} on $P$ processors~\cite{BL98,ABP01,gu2022analysis}.
%We define the model in more detail in the full paper~\cite{fullpaper}.

\myparagraph{Definitions}
We start by introducing definitions used throughout the paper when describing the dynamic connectivity algorithm.
We are given an undirected graph $G$ with vertices $V$ and edges $E$.
Each edge $e = (u,v) \in E$ has a level, $\level{e} \in [1, \lmax]$ assigned to it, where $\lmax = \lceil \log n \rceil$.
Let $E_i = \{e = (u,v) \in E\ |\ \level{e} \leq i\}$ be the set of all edges with levels $\leq i$.
Let $G_i = (V, E_i)$. We maintain the following size invariant on the connected components of each $G_i$:

\begin{invariant}[Size Invariant]\label{inv:size}
The maximum size of a connected component of $G_i$ is $2^{i}$.
\end{invariant}

We can imagine contracting each of the connected components of $G_i$ to obtain $V_i$, the set of \defn{clusters} at level $i$. The \defn{children} of a cluster $c \in V_i$ are the clusters in $V_{i-1}$ that are merged together using level $i$ edges to obtain $c$.

\myparagraph{Cluster Forest}
By explicitly representing the relationship between clusters on different levels, we obtain the \defn{cluster forest}, which we denote using \cf{}, in which each node has a level in $[0, \lmax]$, and where the nodes at level $i$ represent the clusters of connected components of $G_i$ (the graph containing all edges with level ${\leq i}$). 
The root(s) of $\cf$ are nodes representing clusters at level $\lmax$ in $G$ and correspond to the connected components of the graph. 
The leaves of $\cf$ are nodes representing the clusters at level $0$ in $G$, and correspond to the original vertices of the graph.
Figure~\ref{fig:introfig} (b) illustrates the cluster forest for the graph in Figure~\ref{fig:introfig} (a).
If clusters with only a single child are not stored, it is not difficult to see that the number of nodes in \cf{} is $O(n)$.

% Using \cf, checking if two vertices $u$ and $v$ are connected in $G$ amounts to traversing up \cf from the leaves corresponding to $u$ and $v$ to find the roots $r_u$ and $r_v$ of $u$ and $v$, respectively, in \cf; $u$ and $v$ are connected in $G$ if and only if $r_u = r_v$.

Each node $v$ in \cf also stores the size of the cluster it represents, $n(v)$, which is equal to the number of leaves in the subtree rooted at $v$ ($n(v) = 1$ for leaf nodes).
By the size invariant above (Invariant~\ref{inv:size}), for any level $i$ cluster $c$, we have $n(c) \leq \sizeconst{i}$.

For a cluster $c \in \cf{}$ at level $i$, the \defn{cluster graph}, $CG(c)$ of the node is the graph formed by taking its child clusters at level $(i-1)$ as the vertices, and where the edges are the level $i$ edges incident to all leaf vertices in $c$.
Figure~\ref{fig:introfig} (c) illustrates the cluster graph for a vertex in the cluster forest illustrated in Figure~\ref{fig:introfig} (b).

We define a \defn{self-loop edge} as a level $i$ edge for which its endpoints are contained within the same level $(i-1)$ cluster.
\revision{Note that a level $i$ self-loop edge on a level $(i-1)$ cluster $C$ differs from a level $(i-1)$ edge with both endpoints in $C$. The level $i$ self-loop edge appears in the cluster graph of the parent of $C$ as a self-loop on $C$ (thus the name). The level $(i-1)$ edge appears in the cluster graph of $C$ and connects two children of $C$.}

% \subsection{Dynamic Graph Connectivity}
% Our starting point is the HLT algorithm~\cite{holm2001poly}.
% The HLT algorithm distinguishes between {\em tree edges} and {\em non-tree edges}. The tree edges can be any maximum spanning tree over the weighted graph formed by taking the levels of edges to be their weights.
% Each level in their algorithm stores a separate Euler-Tour Tree built over the tree edges in this and all lower levels.

\section{Technical Overview}\label{sec:tech}

\myparagraph{Parallelizing the CF Approach with Low Depth}
%\laxman{Add figure, expand definitions, walk through some example of our approach}
In both the CF and HDT approaches, {\em replacement edge search},
i.e., the search for an edge that certifies the connectivity between the endpoints of a deleted edge (i.e., ``replaces'' it) is the most 
complex aspect of the data structure. 
Both algorithms maintain a hierarchy of nested edge 
subsets $E_1 \subseteq E_2 \subseteq \ldots \subseteq E_k$, with the HDT storing a spanning tree of each subset, and the CF algorithm using a more space-efficient representation (see Figure~\ref{fig:introfig}).
%
%\kuba{At this point it would again be useful to mention that the algorithm maintain a sequence of nested subsets of edges and we refer to edges of $E_i$ as level $i$ edges}
In both algorithms, replacement search is handled by carefully searching the nested hierarchy of components from the lowest level component containing the deleted edge to the highest level. 
For example, in Figure~\ref{fig:introfig}, if the edge $4$--$6$ is deleted, the edges in component $F$ will be searched, and either $5$--$6$ or $5$--$7$ can be used as a replacement edge.
Both algorithms maintain the 
{\em size invariant} (Invariant~\ref{inv:size}) which ensures that components at level~$i$ have
size at most $2^{i}$.
In each component, the non-tree edges
are searched to find a replacement edge, and the unsuccessfully searched
edges are pushed to a lower level to pay for the cost of
searching them.
The parallel version of the HDT algorithm by Acar et al.~\cite{acar2019parallel} obtains parallelism by performing a
doubling search over the non-tree edges incident to the smaller
of the two components induced by the deleted edge (obtaining the smaller component, and indexing the non-tree edges is made possible using Euler Tour Trees).
Note that it is critical that the edges searched are incident to the {\em smaller} component, since this component (and its non-tree edges) 
will be pushed to a lower level to pay for the search.

Unlike the HDT algorithm, which maintains an Euler Tour Tree for every component (see Figure~\ref{fig:introfig} (d)), and can easily split a component into the smaller/larger halves by deleting the edge in the ETT, due to being more space-efficient, the CF algorithm only has access to the {\em cluster graph} of the component that an edge is deleted from. 
Recall that the cluster graph of a level $i$ node consists of {\em all edges} at level $i$ (all such edges go between the level $i-1$ children of this node).
%
%, since the components are stored as explicit nodes in the cluster forest \kuba{I'm not sure what is meant by "storing components as explicit nodes in the cluster forest". I think we should define a cluster graph} (and not redundantly storing all of the nodes in each component at every level) is not easily able to split a component into the ``smaller'' and ``larger'' components induced by the deleted edge.
%
Since the CF algorithm cannot simply remove the edge and split the component into smaller/larger halves, the algorithm performs a more careful {\em graph search} that effectively interleaves two searches (e.g., breadth-first searches) from both level $i-1$ clusters incident to the deleted edge.
Like in the HDT algorithm, level $i$ edges that are unsuccessfully inspected can be paid for by pushing them down to level $i-1$.
We give more details in Section~\ref{sec:seqcf}.
%Instead, the CF algorithm uses a {\em graph search} over the edges
%stored in this cluster-forest node to search for a replacement edge; over the course of the search, the smaller component is discovered and pushed down. 

The graph search, which needs to discover the connected component of a graph undergoing changes, seems extremely difficult to parallelize work-efficiently using the existing CF algorithm.
%\kuba{This may be read as "we tried hard and failed". Can we explain how this is closely connected to basic graph problems which are considered hard to parallelize?}
%
%The key issue (elaborated in Section~\ref{}) is that the CF algorithm needs to {\em discover} the smaller component by searching the edges stored at this level.
%%
%This search can be done with any graph-search procedure starting from the two endpoints of the deleted edge.
%
The main issue is that the graph in the cluster forest at this level
can potentially have very high diameter; in fact, the diameter can
be as high as $\Theta(n)$, making a work-efficient parallelization of this
process very challenging unless we are willing to sacrifice depth.

We solve this problem by introducing a new invariant called the {\em blocked invariant} that
ensures that the cluster graph stored at every internal node in the cluster forest is {\em guaranteed to have constant diameter} (in fact the diameter is always at most $2$).
The key idea of the invariant is simple to state: we ensure that
every cluster is incident to at least one edge that cannot be further
pushed down to a lower level without violating the size invariant.
We prove that this property implies that the cluster graph
at each node is guaranteed to have diameter at most $2$.
As a result, during replacement search any parallel graph
procedure (e.g., parallel breadth-first search) will suffice, and help us obtain low
depth, since the diameter is low.
However, maintaining this property dynamically turns out to be very
tricky even in the sequential setting; our main
algorithmic contributions are novel sequential and parallel algorithms to maintain the invariant under 
(batch-) dynamic updates.
%\kuba{Let's add running time bounds here}
Overall, our new parallel batch-dynamic algorithm can perform insertions and deletions with $O(\log^2 n)$ amortized work per edge update and $O(\log^3 n)$ depth per batch.

%Initially, all the search has access to are $C_u$ and $C_v$, the level $(i-1)$ clusters for $u$ and $v$, but the edge that reconnects $C_u$ and $C_v$ can be between two other level $(i-1)$ clusters that we discover much later on in the search. 
%The problem is that the cluster graph of $P$, the level $i$ parent cluster of $C_u$ and $C_v$, can have high diameter; in fact, the diameter can be as large as the total size of the cluster, $n(P)$. 
%Therefore, the depth of a natural approach to parallelizing the CF algorithm by simply parallelizing the search procedures $\mathsf{search}(C_u)$ and $\mathsf{search}(C_v)$ seems to be lower bounded by the diameter of $CG(P)$.

%As described earlier, the CF algorithm departs from the HDT 
%algorithm in that it explicitly stores the nested connectivity 
%hierarchy as a tree called the cluster forest.
%

%- explain the key challenge (high diameter issue)
%- not at all clear how to directly parallelize the CF approach
%
%to fix the issue, we designed a new invariant. see the paragraph on section 3
%- explain the idea

\myparagraph{The CF Approach in Practice}
In addition to our theoretical contributions, in this paper we give
the first implementations of any cluster forest data structure, and study
the practicality of the CF approach and our new invariants in the 
sequential setting.
As we discuss in more detail in Section~\ref{sec:exp}, implementing CF algorithms
seems to be even more involved than implementations of the HDT algorithm.
A significant implementation challenge is that a single data 
structure---the cluster forest---stores
both the hierarchy of connected components and the non-tree edges stored at
each component in the hierarchy (in the HDT algorithm, the implementation complexity is
somewhat lower since each component is stored as a separate Euler Tour Tree).

Our experiments show that our new implementations of the CF approach
are significantly more space-efficient and process updates faster than existing state-of-the-art dynamic connectivity implementations.
For example, across a diverse set of graph inputs, our CF implementations
achieve up to 19.7$\times$ lower space usage compared with a carefully 
designed implementation of HDT, and up to 6.2$\times$ faster updates.

% find . -name "*.cc" | xargs wc -l

%\kuba{It would be good to mention that a parallel implementation is an open problem / future work. It would be good to somehow say that this would require more than just coding and list specific challenges.}

%Practical contributions (notes):
%\begin{itemize}
%   \item first study of the CF approach; first implementation of any cluster forest data structure.
%
%   \item show that it achieves good space usage in line with the theory, assuming we carefully perform compression.
%
%   \item show that the CF approach achieves better update speeds although they are the same in theory
%   
%   \item We also show queries that are equivalent in performance to HDT based implementations
%
%   \item Practical optimization: a practical and theory-preserving optimization that improves the practical performance of the algorithm (inserting to LCA and inserting to blocked). The original
%   paper only inserted to the root.
%
%\end{itemize}
%
%Space:
%HDT looks like 4(nlogn); ours looks like 2(nlogn) (only n leafs in the CF tree)
%ETT has 2*n nodes for a tree with n vertices; extra overhead.
%
%Node height:

%\myparagraph{Our Contributions}
%To summarize, the main contributions of this paper are:
%\begin{enumerate}
%   \item A, B, C 
%\end{enumerate}

%{\bf perform a deeper theoretical and practical investigation
%of properties of the CF approach}.
%
%We find that 

%
%
%%
%Another challenge is their poor scalability due to lots of indirection
%(pointer chasing) and thus their typically poor locality.
%%
%For example, the classic HDT algorith
%
\section{Sequential CF Algorithm}\label{sec:seqcf}

%The starting point for our work is the sequential algorithm of Holm, de Lichtenberg, and Thorup~\cite{holm2001poly}, which we refer to as the {\em HDT algorithm}.
%The HDT algorithm can be adapted to use linear space using ideas first described by Thorup~\cite{thorup2000near} and later by Wulff-Nilsen~\cite{wulff2013faster}. We refer to this algorithm as the {\em Cluster Forest} or {\em CF algorithm}.
Next, we review how the sequential CF algorithm performs insertions and deletions.

\myparagraph{Insert $e=(u,v)$} 
The CF algorithm first sets the level of the edge, $\level{e} = \lmax$. Let $r_u$ and $r_v$ be the roots of the trees containing $u$ and~$v$, respectively, in \cf.
If $r_u = r_v$, then nothing further needs to be done.
If $r_u \neq r_v$, then we increased the connectivity of the graph and must merge $r_u$ and $r_v$ together.
This requires a {\em merge} operation on \cf{} that takes the roots of two trees and merges them together by adding the children of (without loss of generality) $r_v$ as children of $r_u$ and deleting $r_v$.

\myparagraph{Delete $e=(u,v)$} As in the HDT algorithm, deletion requires performing a {\em replacement edge search} to check if the deletion of $(u,v)$ affects the connectivity of $G$.
Let $i = \level{e}$.
In the CF algorithm, the problem boils down to {\em certifying the connectivity} of the level $i$ cluster\revision{, $P$,} containing $u$ and $v$.
\revision{Recall that the cluster graph of $P$ consists of the level $(i-1)$ clusters that are children of $P$ and the level $i$ edges with both endpoints in $P$.}
Let the level $(i-1)$ clusters containing $u$ and $v$ be $C_u$ and $C_v$, respectively, and the cluster graph of $P$, $CG(P)$, be $G_i$.
If $C_u = C_v$ ($e$ is a self-loop), then we can quit since the connectivity of $G_i$ is unaffected after deleting $e$.

If $C_u \neq C_v$, then we need to check whether $(u,v)$ was a {\em bridge} of $G_i$ (i.e., every possible spanning tree of $G_i$ must use this edge). If  the deleted edge is a bridge, then $G_i$ splits into two components. This means that the cluster forest \cf must be updated, and then the algorithm must recursively check at level $(i+1)$ whether the two split pieces of $G_i$ can be reconnected using a level $(i+1)$ edge.

\begin{figure}
\centering
\includegraphics[width=0.48\textwidth]{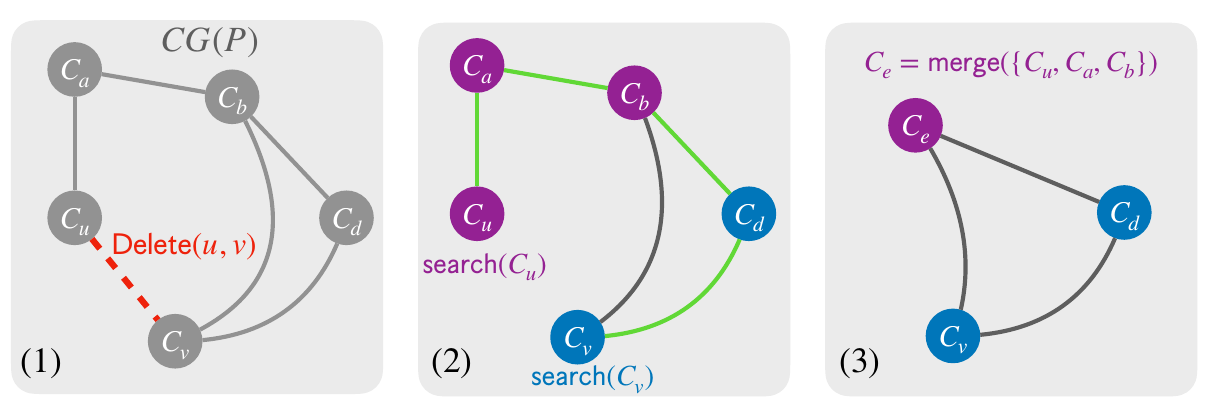}
\caption{\label{fig:wndeletion}
\small{}
    {\bf (1):} The cluster graph of the level $i$ cluster $P$ containing a deleted level $i$ edge $(u,v)$; the vertices of the cluster graph are the level $(i-1)$ child clusters of $P$.
    {\bf (2):} Deleting $(u,v)$ may disconnect $CG(P)$, so $\mathsf{search}(C_u)$ and $\mathsf{search}(C_v)$ are run to check if $C_u$ and $C_v$ are still connected using level $i$ edges in $CG(P)$. Green edges are edges explored during the search.
    The size of all clusters explored by the smaller search must have size $\leq \sizeconst{i-1}$, and can be merged into a single level $(i-1)$ cluster.
    {\bf (3):} In this example, $CG(P)$ remains connected, $\{C_u,C_a,C_b\}$ are merged, and the explored level $i$ edges of the smaller search are pushed to level $(i-1)$.
    %\julian{if we have space, it would be helpful to have an example of the other case as well.}
    }
\end{figure}

To certify connectivity in $G_i$, the CF algorithm runs two searches from both $C_u$ and $C_v$ using any graph search procedure, e.g., depth-first search. Call these searches $\mathsf{search}(C_u)$ and $\mathsf{search}(C_v)$. The searches explore the multi-graph of level $(i-1)$ clusters that are children of $P$, and level $i$ edges between them (see Figure~\ref{fig:wndeletion}).
%
%In the HDT algorithm, the Euler-Tour trees let us quickly identify which of $C_u$ and $C_v$ is in the smaller component and only process the edges in the smaller component which can all be pushed down a level, thus paying for their inspection.
%
Unlike in the HDT algorithm, the CF algorithm doesn't know which of $C_u$ or $C_v$ is in the smaller component after deleting $(u,v)$.
To obtain amortized work bounds, the CF algorithm alternates between running steps of $\mathsf{search}(C_u)$ and $\mathsf{search}(C_v)$. 
The searches stop when either (1) a common vertex in $G_i$ is explored by both searches, or
(2) one of the searches runs out of level $i$ edges to explore, certifying that $C_u$ and $C_v$ are no longer connected using level $i$ edges.

In case (1), $C_u$ and $C_v$ must be connected so $G_i$ is unaffected, and thus the cluster forest \cf does not change. 
Let $S_u$ and $S_v$ be the set of level $(i-1)$ clusters explored by $u$ and $v$'s searches, respectively. Let $n_u = \sum_{c \in S_u} n(c)$ and $n_v = \sum_{c \in S_v} n(c)$. 
Then $\min(n_u, n_v) \leq \sizeconst{i-1}$, and so we can push all of the level $i$ edges explored by the smaller search down to level $(i-1)$. The alternating search ensures that we don't need to worry about pushing down edges incident to the larger component as the work on these edges is paid for by the pushing down of edges in the smaller component. Pushing edges to level $(i-1)$ can require merging level $(i-1)$ clusters.

In case (2), we take the search with smaller total size value, with out loss of generality $\mathsf{search}(C_u)$, and push all of the level $i$ edges it explored down to level $(i-1)$ to pay for the search. This requires merging all level $(i-1)$ clusters explored by $\mathsf{search}(C_u)$. Let $W$ be this new level $(i-1)$ node that they were all merged into. 
Next, we need to update \cf to reflect the fact that $G_i$ split. We first remove $W$ as a child of $P$ in \cf and decrement $n(P)$ by $n(W)$.
% \julian{does $n$ take as function a cluster or a vertex?} \laxman{It can take any node in the cluster tree (the leaves of the cluster tree are vertices)---clarified when defining $n(v)$ above}. 
We then create a new level $i$ node $P'$ and set the parent of $W$ to $P'$, and add $P'$ as a child of the parent of $P$.
Only in case (2) do we need to continue (recursively) at level $(i+1)$ to check if the level $i$ clusters containing $u$ and $v$ remain connected using the level $(i+1)$ edges, or whether a similar split needs to happen at level $(i+1)$ or a higher level.

\myparagraph{Cluster Forest Interface}
Studying this algorithm, we can identify the following necessary operations on the cluster forest \cf{}:
\begin{enumerate}[label=(\arabic*),topsep=0pt,itemsep=0pt,parsep=0pt,leftmargin=20pt]
\item \textbf{FetchEdge}$(C,i)$ the search needs to be able to iterate over the level $i$ edges incident to a level $(i-1)$ cluster $C$.
\item \textbf{Parent}$(C)$ returns the parent of $C$ in the cluster forest.
\item \textbf{AddChild}$(P,C)$ adds a node $C$ as a child of node $P$.
\item \textbf{RemoveChild}$(P,C)$ removes $C$ from the children of $P$.
\item  \textbf{Cluster}$(v, l)$ returns the level $l$ cluster containing $v$.
\item \textbf{Merge}$(C_1, C_2)$ merges two level $(i-1)$ clusters $C_1$ and $C_2$.
\item \textbf{PushDown}$(e)$ pushes down a level $i$ edge $e$.
\end{enumerate}

%Really, Item~\ref{item:graph} is an operation on the graph, and should be formalized using a graph-level API.
%We will formalize both a graph-level ADT and an ADT for \cf{} in Section~\ref{sec:interfaces}.

\myparagraph{Local Trees Implementation}
Since a node in the cluster forest may have at most $n$ children, the CF algorithm represents the cluster graph of each cluster in \cf using {\em local trees} to allow performing all the above operations in $O(\log{n})$ time. 
A local tree for a cluster $u$ in \cf is a binary tree $L(u)$ where the children of $u$ in \cf are leaves of $L(u)$. 
Let the rank of a child $v$ of $u$ be $r(v) = \lfloor \log n(v) \rfloor$. 
Initially, each child cluster is its own tree.
While there are two trees $r$ and $r'$ of the same rank, they are paired up into a new tree $r''$ with rank one larger.
Once this pairing process terminates, there are at most $\log n$ trees $T_1, \ldots T_k$ \revision{which are called \defn{rank trees}. The collection of rank trees} are combined into a single binary tree by connecting them in order of rank along the right spine of a binary tree.

To efficiently iterate over level $i$ edges during the search, the trees are augmented using a $\log n$-length {\em bitmap} (stored as a single word) where the $i$-th bit is $1$ if and only if there is a level $i$ edge incident to some leaf vertex in the subtree. 
%This augmentation is necessary to efficiently iterate over the level $i$ edges incident to a level $(i-1)$ cluster during the searches.
%
All of the operations needed in the sequential CF algorithm can be implemented in $O(\log n)$ worst-case time using local trees~\cite{wulff2013faster}.

\myparagraph{Advantages and Challenges of the CF Algorithm}
One immediate advantage of the cluster forest algorithm~\cite{wulff2013faster} is that the space requirement for \cf can be made $O(n + m)$ by simply ensuring that the cluster forest \cf{} is path compressed, i.e., a level-$i$ cluster $c$ is explicitly represented if and only if either $i=0$ or there is a level $i$ edge $e$ with both endpoints in $c$ (thus either $e$ is a self-loop or $c$ has at least two level-$(i-1)$ child clusters).
The {\em compressed} designation means that we apply path compression on the data structure. That is if a level $i$ node isn't incident to any level $i$ edges (thus it has only one child), then that node is not explicitly represented in the data structure.
On the other hand, the HDT algorithm requires $O(n\log n + m)$ space, since each vertex is potentially present in an Euler Tour Tree at all $O(\log n)$ levels. Since large graphs in practice are often extremely sparse~\cite{dhulipala2018theoretically}, achieving linear total space is an important goal that can lead to practical and theoretically-efficient implementations of dynamic graph connectivity.
%In principle the practical space usage of a faithful implementation should be reasonable, although we will experimentally see whether this is the case.

As discussed in Section~\ref{sec:tech}, the main challenge with parallelizing the CF algorithm is how to perform the replacement edge search, which requires work-efficiently traversing the cluster graph (which can potentially have very high diameter, even diameter $\Theta(n)$).

\section{The Blocked Cluster Forest}

We will use the idea of {\em blocked edges} to obtain more structured cluster graphs that have bounded diameter and enable us to parallelize the connectivity search. \revision{Here we define blocked edges and establish the main invariant of our new data structure:}

\begin{definition}[Blocked Edge]\label{def:blocked_edge}
A level $i$ edge $e=(u,v)$ is a \defn{blocked edge} if it cannot be pushed to level $(i-1)$ without violating Invariant~\ref{inv:size}. I.e., a level $i$ edge $(u,v)$ is blocked if and only if $n(\cluster{u,i})+n(\cluster{v,i}) > \sizeconst{i-1}$.
An \defn{unblocked edge} is an edge that is not blocked.
\end{definition}

\begin{invariant}[Blocked Edge Invariant]\label{inv:blocked}
Consider any cluster graph $CG(c)$ of a level $i$ cluster $c$. Then every level-$(i-1)$ cluster $X \in CG(c)$ is incident to at least one blocked level -$i$ edge or $c$ is an isolated cluster and only has one child $X$. 
\end{invariant}

A \defn{blocked cluster-forest} is defined as a cluster forest where every cluster satisfies Invariant~\ref{inv:blocked}.
An \defn{isolated cluster} in the cluster forest is a cluster that has only a single child.
The blocked edge invariant is useful since we can show that it implies that every cluster graph has low diameter, making them more amenable to parallel search. 
\revision{Next we describe several important and useful properties of the blocked cluster-forest. The proofs are left to the appendix.}
The key property is that a maximum matching computed over the blocked edges must have size at most $1$, as summarized by the following lemma:

\begin{restatable}{lemma}{blockedmatchingsize}\label{lem:blocked_matching_size}
    Suppose Invariant~\ref{inv:blocked} holds for a cluster graph $CG(c)$ of a level $i$ cluster $c$. Let $M$ be the size of the maximum matching in $CG(c)$ over only the blocked edges. Then $M \leq 1$.
\end{restatable}

The invariant also implies that we cannot have a path using only blocked edges of length $\geq 3$, since such a path has a blocked matching of size $M > 1$. This allows to prove the following lemma that bounds the diameter of any cluster graph.

\begin{restatable}{lemma}{diameter}\label{lem:diameter}
Suppose Invariant~\ref{inv:blocked} holds for a cluster graph $CG(c)$ of a level $i$ cluster $c$.
Then $CG(c)$ has diameter $\mathsf{diam}(CG(c)) \leq 2$.
\end{restatable}

The following two lemmas describe other properties of the blocked cluster forest which enable efficient updates:

\begin{restatable}{lemma}{radius}\label{lem:radius}
Suppose Invariant~\ref{inv:blocked} holds for a graph $CG(c)$ of a level $i$ cluster $c$. Then there exists a {\textbf{center}} node in $CG(c)$ that is connected to every other node by a blocked edge. 
%Moreover, if $CG(c)$ has at least $4$ nodes, then the center node is uniquely defined.
\end{restatable}

%\laxman{Perhaps collapse the following lemma into the previous one, or remove the last sentence of the previous lemma?}
\begin{restatable}{lemma}{heavycenter}\label{lem:heavy_center}
Suppose Invariant~\ref{inv:blocked} holds for a graph $CG(c)$ of a level $i$ cluster $c$, and $CG(c)$ has $k \geq 4$ nodes. Then the center node corresponds to the largest cluster in $CG(c)$.
\end{restatable}

\myparagraph{Using Structure for Parallel Connectivity Search}
Consider a level $i$ cluster $c$. 
We can characterize the cluster graph for a cluster~$c$ as one of following:

\begin{enumerate}
\item The cluster graph is a star (case (1) in Figure~\ref{fig:blocked_structure}).
\item The cluster graph is a triangle (case (2) in Figure~\ref{fig:blocked_structure}).
\item The cluster graph is a single node (case (3) in Figure~\ref{fig:blocked_structure}). 
%In this case, the parent cluster is isolated.
\end{enumerate}

\begin{figure}
\centering
\vspace{-0.5em}
\includegraphics[width=0.4\textwidth]{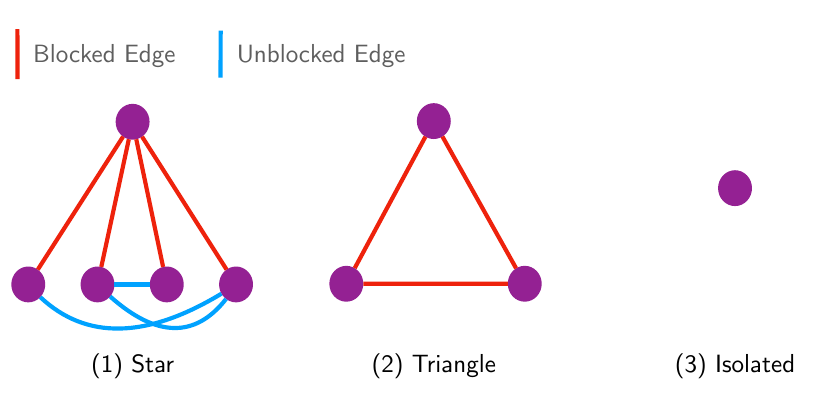}
\caption{\label{fig:blocked_structure}
\small
    The three possible cases for what the cluster graph of a level $i$ cluster in a blocked cluster-forest can look like. Red edges are blocked edges, and blue edges are unblocked edges.}
\end{figure}

Figure~\ref{fig:star} illustrates the relationship between clusters in the cluster forest \cf{} and the cluster graph of a node.
Since the graphs are guaranteed to have low diameter, performing a graph search on $CG(c)$ from two clusters $C_u$ and $C_v$ can be done in low depth by running the searches in lock step and doubling the number of edges that we explore at each step. 
The doubling ensures that we can still amortize the exploration cost to level decreases on the edges in previous steps while ensuring that the search runs in $O(\mathsf{polylog}(n))$ depth.
The main challenge now is how to maintain the blocked invariant dynamically.

\begin{figure}
\centering
\includegraphics[width=0.47\textwidth]{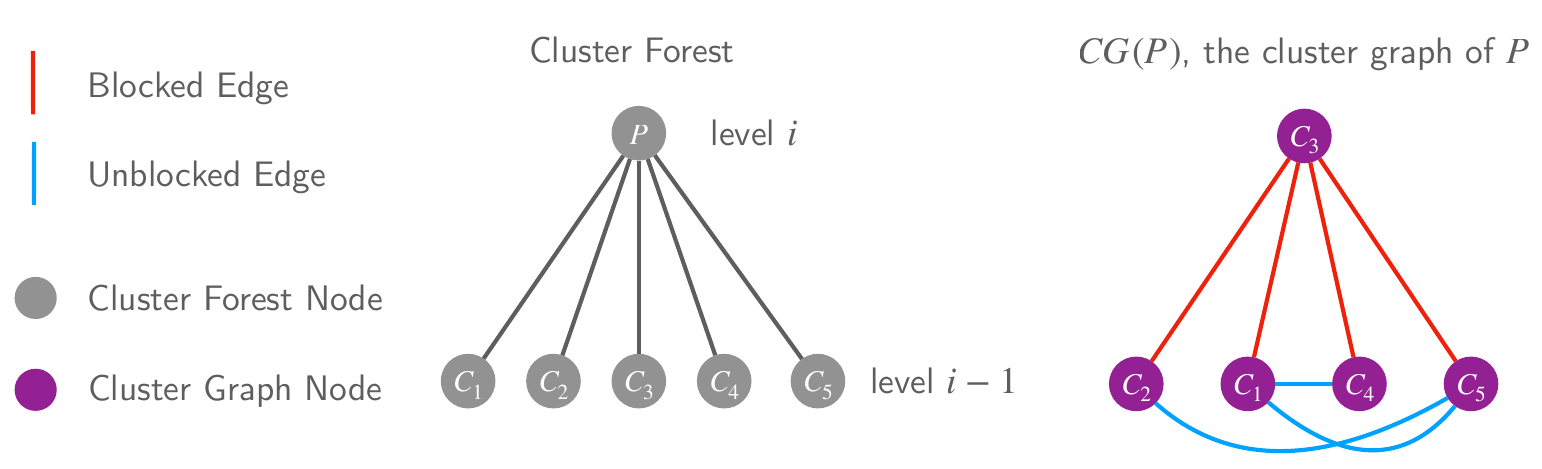}
\caption{\label{fig:star}
\small
Illustration of the cluster graph for a node in the blocked cluster-forest, and its star structure. The center node of $CG(P)$ is $C_3$, and all other nodes are satellites connected to the center through a blocked edge.}
\end{figure}

\subsection{Updating the Blocked Cluster-Forest}
In this section, we describe how to maintain Invariant~\ref{inv:blocked} while performing single edge insertions and deletions. 
Surprisingly, achieving this goal even in the sequential setting requires some non-trivial ideas and analysis.
%Section~\ref{sec:batchpar} describes how to extend our ideas to the parallel batch-dynamic setting.

\myparagraph{Pushing Edges Down} Note that pushing an edge from level $(i+1)$ to $i$ may violate the blocked invariant for its new level $(i-1)$ endpoints.
Lemma~\ref{lem:unblocked_isolated} and \ref{lem:push_down_until_blocked} prove that if an unblocked edge is {\em repeatedly} pushed down until it is blocked, then the blocked invariant will be maintained. \revision{The proofs are left to the appendix due to space constraints.}
% From here on out when we say we will push an edge down or merge two clusters we mean that we will push the edge down until it is blocked. The work can still be charged to that edge for every level it is pushed down.

\begin{restatable}{lemma}{unblockedisolated}\label{lem:unblocked_isolated}
    If a level $i$ edge $e$ between two distinct level $(i-1)$ clusters $C_1$ and $C_2$ is unblocked, then either $C_1$ or $C_2$ is an isolated cluster, i.e., its cluster graph consists of a single level $(i-2)$ cluster.
\end{restatable}

\begin{restatable}{lemma}{pushuntilblocked}\label{lem:push_down_until_blocked}
    If an unblocked edge is pushed down until it is blocked then Invariant \ref{inv:blocked} is preserved.
\end{restatable}

\myparagraph{Insert $e = (u,v)$}
We give a simple top-down algorithm for insertion. 
In the cluster forest we include a single global level \revision{$\lmax+1$} cluster whose children are the roots of all the components. 
Then for insertion, add the edge as a level $(\lmax+1)$ edge.
% , and then push it down as far as possible. 
No level $(\lmax+1)$ edge can be blocked as the size constraint for level $\lmax$ clusters is $\geq n$. 
% The newly inserted $e$ is now a level $(\lmax+1)$ edge between two level $\lmax$ clusters or it is a self loop. 
Next repeatedly push down $e$ until it is blocked. 
Lemma~\ref{lem:push_down_until_blocked} proves that this maintains the blocked invariant.

\myparagraph{Delete $e = (u,v)$}
Let the level of $e$ be $i = \level{e}$ prior to deletion.
Similar to before, let $C_u = \cluster{u, i-1}$ and $C_v = \cluster{u, i-1}$, the two level $(i-1)$ clusters containing $u$ and $v$. Let $P$ be the parent cluster of $C_u$ and $C_v$, and $GP$ be the grandparent. Let $G_i=CG(P)$ be the cluster graph of $P$.
Like before, if $C_u = C_v$ then the edge is a self-loop we can quit here since the connectivity is unaffected.
If the edge is unblocked it can be safely deleted because either it is a self loop or the clusters are connected by 1 or 2 blocked edges.
If $C_u \neq C_v$ and the edge is blocked, deletion of the edge may cause $C_u$ and $C_v$ to be in violation of the blocked invariant.

Consider a level $i$ edge.
Any unblocked level $i$ edges that we encounter can be pushed down (repeatedly until blocked).
Any blocked level $i$ edges must remain at this level. Since there can be multiple (parallel) edges between two level $(i-1)$ clusters that are all blocked, we need to be careful that we don't process too many blocked edges, since we can't push these down.

To restore the blocked invariant at this level (the lowest level that the deletion affects), we repeatedly fetch and push down a level $i$ edge incident to $C_u$ or $C_v$ in an alternating fashion until we have certified that both are incident to a blocked edge (or have no more incident edges).
Restoring the invariant at higher levels turns out to be more involved, and we describe this shortly.

Once we have restored the invariant at this level, we can certify the connectivity of $C_u$ and $C_v$ as follows. If both have no more incident edges, then they are connected if and only if $C_u=C_v$. 
If one has no more incident edges and the other is incident to a blocked edge, they must not be connected.
If both are incident to a blocked edge, they must be connected (otherwise the two blocked edges form a matching of size 2 over the blocked edges in $CG(P)$).

Just as in original CF algorithm, if the connectivity search failed at level $i$, the algorithm must continue to level $(i+1)$ and perform connectivity search over the level $(i+1)$ edges.
In this case, one of the components must have merged into a single level $(i-1)$ cluster. Let $W$ be this cluster. We (1) remove $W$ as a child of $P$ and decrement $n(P)$ by $n(W)$, (2) create a new level $i$ cluster $P'$, set the parent of $W$ to $P'$, and (3) add $P'$ as a child of $GP$.

\revision{When $W$ is removed as a child of $P$ and added as a child of a new cluster $P'$, $P$ is essentially being split into two clusters, $P_1$ and $P_2$, in the cluster graph of $GP$ (Figure~\ref{fig:center_split} shows an example of a cluster graph before and after a cluster is split).}
Then the size of $P$ is split between $P_1$ and $P_2$: $n(P) = n(P_1) + n(P_2)$.
Any edges previously incident to $P$ are now incident to either $P_1$ or $P_2$ (or possibly both if it was a self-loop on $P$).
Since the sizes of both $P_1$ and $P_2$ are less than that of $P$, these edges may now be unblocked, thus the cluster graph of $GP$ potentially violates the blocked invariant. 

\begin{figure}
\centering
\includegraphics[width=0.32\textwidth]{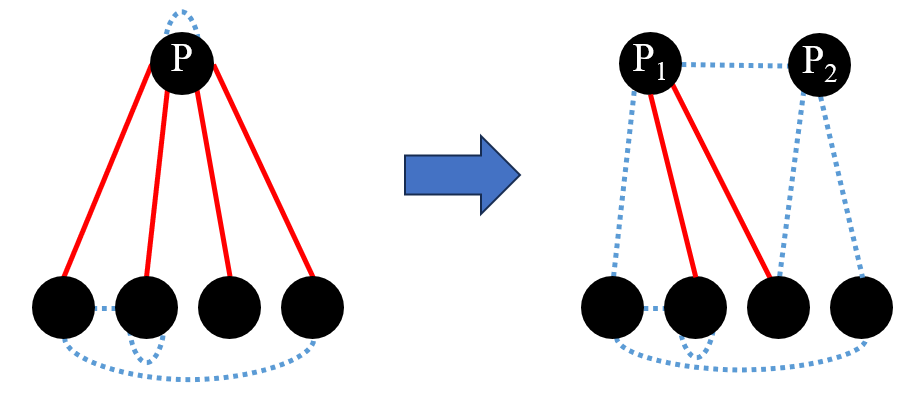}
\caption{\label{fig:center_split}
\small
    \revision{
    A cluster graph before and after the center cluster $P$ is split. 
    The edges incident to $P$ are now split between $P_1$ and $P_2$. To restore the blocked invariant after splitting the center, the algorithm carefully fetches edges out of each satellite cluster.}
}
\end{figure}

\myparagraph{Restoring the Blocked Invariant at Higher Levels}
Unlike restoring the blocked invariant at the lowest level that an edge was deleted, which is  straightforward, restoring the invariant at higher levels is much trickier since clusters can be split as illustrated in Figure~\ref{fig:center_split}.
If the number of clusters in $CG(c)$ is bounded by some constant, we can afford to restore the invariant by checking that every single cluster is incident to a blocked edge, pushing any found unblocked edges down.

Otherwise (if the number of clusters is $\omega(1)$), Lemmas~\ref{lem:radius} and \ref{lem:heavy_center} tell us that there is a well-defined center cluster that is connected to every other cluster by a blocked edge, and this cluster has the largest size in the cluster graph.
Thus we can naturally think of two cases: $P$ was the center cluster or $P$ was a satellite cluster.
We define a \defn{satellite cluster} as a cluster that is not the center, it was connected to the center via a blocked edge before the split of $P$.

If $P$ was a satellite cluster that split into fragments $P_1$ and $P_2$, then only those two clusters could violate the blocked invariant. We can restore the invariant by simply fetching and pushing down one \defn{outbound edge} incident to both fragments. 
We define an outbound edge as an edge incident to a cluster that is not a self-loop. 
If the outbound edge is blocked, then we are done. If it is unblocked, we push it down thus merging the fragment into another cluster that is incident to a blocked edge.

\revision{\myparagraph{Restoring the Blocked Invariant With Split Center}}
When $P$ is the center cluster, the blocked invariant may become violated for $P_1$ and $P_2$ themselves, as well as many or all satellite clusters.
\revision{The difficulty in restoring the blocked invariant in this case is that we have to restore the blocked invariant for potentially many satellites, but we cannot examine too many blocked edges since they cannot be pushed down to charge the work.}
The structure of this new cluster graph consists of the two centers, $P_1$ and $P_2$, and a set of satellite clusters connected by edges to $P_1$ and/or $P_2$. There may also be edges between $P_1$ and $P_2$ (previously self-loops on $P$), edges between satellites, and self-loops on any cluster.
Figure~\ref{fig:center_split} shows a possible cluster graph before and after the center splitting.
\revision{Here we loosely describe how to restore the blocked invariant after a center has been split. We leave a more careful description and analysis to the appendix due to space constraints. Our result is summarized by the following lemma:}

\begin{lemma}\label{lem:seq_restore_inv}
    When a cluster $X$ is split in a cluster graph $CG(c)$ that previously maintained Invariant~\ref{inv:blocked}, the invariant can be restored in $O((k+1) \log n)$ time where $k$ is the total number of edges pushed down by this process.
\end{lemma}

\revision{
The high-level procedure is to sequentially fetch an outbound edge incident to each satellite cluster.
We maintain a running total of the sizes of $P_1$ and $P_2$ combined with the sizes of their adjacent satellites that have been discovered and sets of the adjacent satellites found. We update these total sizes when the fetched edge connects to one of the center fragments. If it connects to another satellite, the two satellites merge together (the edge may be blocked, but we proved in the appendix that this can only happen once).
This continues until either (1) an edge from every satellite has been found, or (2) $P_1$ or $P_2$ has found two neighboring satellites that could not merge with it due to the size constraint (e.g. the two edges to the center fragment would be blocked if all previously discovered satellites were merged with their center).
In the first case it is easy to restore the blocked invariant by pushing down all but at most two of the edges found to the center fragments.
In the second case, assume without loss of generality that this happened with $P_1$. Then it is certain that $P_2$ is able to merge with all of its neighboring satellites without violating the size constraint. Thus we can fetch all of the edges out of $P_2$ and merge it with all of its neighbors to restore the blocked invariant.
}
%Now we describe how to restore the blocked invariant when a star center is split. The high-level procedure is to fetch an outbound edge incident to each satellite cluster.

Once the invariant has been restored at a level, the connectivity between two clusters can be certified just as described before. Then either the deletion is complete, or we have to split the parent cluster and continue searching edges at the next level up.

\revision{
\myparagraph{Analysis}
Lemma~\ref{lem:deletion_correct} proves the correctness of our deletion algorithm in maintaining the size invariant and the blocked edge invariant, the formal proof of which is described in the appendix.
The overall correctness of the algorithm follows from this and the original cluster forest algorithm.
}

\begin{restatable}{lemma}{deletioncorrect}\label{lem:deletion_correct}
    Invariants~\ref{inv:size} and~\ref{inv:blocked} are preserved by edge deletion.
\end{restatable}

Deletions may take $O(\log n)$ work per level of the cluster forest that is not charged to any edge being pushed down if a blocked edge is immediately encountered when certifying connectivity. This results in a total of $O(\log^2n)$ uncharged work. Combining this with Lemma~\ref{lem:seq_restore_inv} yields the following theorem:

\begin{theorem}
The amortized cost of insert and delete operations is $O(\log^2 n)$ using the blocked cluster-forest data structure.
\end{theorem}

\section{Parallelizing Individual Updates}\label{sec:parallelupdate}
Next, we step towards our goal of a parallel batch-dynamic algorithm by describing how to perform a single update in parallel. 
Many ideas and primitives developed in this setting are also useful in the batch-dynamic setting (Section~\ref{sec:batchpar}).
The depth of insertion is $O(\log^2 n)$ already, so we focus on deletion.

\subsection{Batch-Dynamic Local Trees}\label{sec:modlocaltree}
We use a modified version of the local trees used in ~\cite{wulff2013faster} that also supports parallel operations.
\revision{Additionally our parallel algorithm requires an operation to return a prefix of the smallest clusters sorted by size in the local tree.}
The interface for the modified local trees is defined as follows:
\begin{itemize}[topsep=0pt,itemsep=0pt,parsep=0pt,leftmargin=15pt]
    \item \textbf{BatchInsert}$(c_1,...,c_k)$ takes a list of $k$ new clusters and adds them to the local tree.
    \item \textbf{BatchDelete}$(c_1,...,c_k)$ takes a list of $k$ clusters in the local tree and removes them.
    \revision{
    \item \textbf{GetMaximalPrefix}$(s)$ returns a maximal prefix of the elements in the local tree sorted by size such that their total size is less than or equal to $s$.
    }
\end{itemize}

\revision{To support these operations efficiently we divide the elements in the tree into $O(\log n)$ size classes of geometrically increasing intervals, and store a separate weight-balanced tree for the clusters in each size class.
In the weight-balanced trees, elements are keyed by the size of the cluster they represent. This allows us to efficiently return a prefix of the elements. Since the sizes of all elements are within a constant fraction of each other, the scoping sum argument of~\cite{wulff2013faster} still applies.
Additionally, since weight-balanced trees are known to support efficient batch-parallel insertion and deletion~\cite{blelloch2022joinable}, our local trees can support this by first semi-sorting by size class and then separately and in parallel batch inserting and deleting elements into/from the weight-balanced tree for each size class.}
\revision{All operations} can be done in $O(\log n)$ depth and $O(k\log(1+n/k))$ work. More detail and analysis of this data structure are described in the appendix.

\subsection{Searching for Edges in Parallel}
To achieve low depth in batch updates we need an operation to fetch $k$ level $i$ edges incident to a given cluster in low depth:
\begin{itemize}[topsep=0pt,itemsep=0pt,parsep=0pt,leftmargin=15pt]
    \item \textbf{FetchEdges}$(k,C,i)$ returns $k$ level-$i$ edges incident to a cluster $C$, or all of the level-$i$ edges if there are less than $k$.
\end{itemize}

To implement this we traverse down the nested local tree structure one level at the time, maintaining a set of approximately $k$ clusters that have a level $i$ edge incident to them by following the set bits in the bitmaps. The algorithm is described in detail in the appendix. This yields the following lemma:

\begin{restatable}{lemma}{batchfetch}\label{lem:batch_fetch}
    Fetching $k$ level-$i$ edges incident to a given cluster (or all edges if $<k$ exist) in a cluster forest implemented with local trees can be done in $O(\log n)$ depth and $O(k'\log n)$ work where $k'$ is the number of edges returned.
\end{restatable}

Now we describe two useful routines to fetch edges in parallel which will be used in the rest of this section:
\begin{enumerate}
    \item Fetch all edges incident to a cluster
    \item Fetch a single outbound edge incident to a cluster
\end{enumerate}
Note that fetching a single outbound edge is not trivial because it is possible to repeatedly find self-loops when fetching an edge incident to a cluster.
The goal is to implement both of these operations in $O(\log^2 n)$ depth and work within a constant factor of the work of their sequential implementations (e.g. $O(k\log n)$ where $k$ is the total number of edges fetched).

To accomplish this we perform a doubling search where we fetch edges in rounds, doubling the number of edges fetched in each round. 
In each round we use the parallel \textbf{FetchEdges} operation which fetches $k$ edges incident to a cluster in $O(\log n)$ depth and $O(k \log n)$ work. 
This process will take $O(\log n)$ rounds of doubling since there are at most $O(n^2)$ edges.
After each round we will check if the search is complete. 
For the first operation, we know it is done if \textbf{FetchEdges} returns $<k$ edges. 
For the second operation, we check all of the edges returned in parallel to see if they are a self-loop.
This process takes $O(\log n)$ depth and $O(\log n)$ work per edge.
For both operations we have fetched at most twice as many edges as the sequential implementation, and thus performed work within a constant factor of the sequential work. 
Each of the $O(\log n)$ rounds of doubling search takes $O(\log n)$ depth, giving a total depth of $O(\log^2 n)$. 
This yields the following Lemma:

\begin{lemma} \label{lem:doubling_fetch}
    Fetching all of the edges incident to a cluster or fetching an outbound edge incident to a cluster can be done in $O(\log^2 n)$ depth and $O(k \log n)$ work where $k$ is the total number of edges fetched.
\end{lemma}

\subsection{Pushing Down Edges in Parallel}
\revision{
\myparagraph{Updating Bitmaps in Parallel}
First we describe the simpler problem of updating the bitmaps in the local trees for a batch of edges that are pushed down from the same level.
Updating the bitmaps is equivalent to updating a batch of $k$ augmented values in the nested local trees structure and can be implemented using an atomic compare-and-swap (CAS) operation on each bit along the leaf-to-root path that needs to be updated, stopping if the CAS fails. In total the batch bitmap update takes $O(\log n)$ depth and ${O(k\log(1+n/k))}$ work.
}

\revision{\myparagraph{Pushing Down Groups of Edges in Parallel}}
Here we describe how to handle pushing edges down in parallel. We provide the following two routines:
\begin{itemize}[topsep=0pt,itemsep=0pt,parsep=0pt,leftmargin=15pt]
    \item \textbf{PushDownGroup}$(E)$ pushes down a set $E$ of $k$ level $i$ edges, where the edges in $E$ are all incident to a common level $(i-1)$ cluster and the total size in level $(i-1)$ clusters containing their endpoints is $\leq \sizeconst{i-1}$ (e.g. all of $E$ can be pushed down).
    %\adam{So $E$ are all unblocked, right?} \quinten{They are definitely all unblocked, but this requirement is stronger: that they can all be pushed down. Pushing down one unblocked edge may cause another one to be blocked meaning they can't all be pushed down even if they were all unblocked before.}
    \item \textbf{BatchPushDown}$(E)$ given a set $E$ of level $i$ edges, this pushes down as many edges of $E$ as possible until the only remaining ones are blocked. Formally, this operation ensures that every level $(i-1)$ cluster containing an endpoint of any edge in $E$ is incident to a blocked edge or its parent is isolated.
    %\adam{Are edges $E$ all unblocked? I guess no -- so it's confusing because PushDownGroup assumes that} \quinten{They are not. I changed the description a little to describe the intuition behind this operation.}
\end{itemize}

Pushing down a single level $i$ edge $e$ to level $(i-1)$ requires updating the level $i$ and $(i-1)$ bitmaps, and possibly combining the local trees of the two level $(i-1)$ clusters containing the endpoints of $e$. 
In the sequential setting, our method to maintain the blocked invariant and the size invariant after a single edge push was to push down an edge as far as possible until it became blocked. 

When pushing down multiple edges in parallel there are a few challenges that arise. There may be situations where pushing down one edge causes a different previously unblocked edge to become blocked. Also, combining several local trees in parallel is non-trivial.

\myparagraph{Our Approach: Reducing to \textbf{PushDownGroup}}
Consider pushing down a batch of $k$ edges incident to a common cluster.
\revision{Updating the bitmaps for this batch of edges can be done efficiently as described at the beginning of this section.}
Lemma~\ref{lem:push_down_group} proves that when pushing down a group of unblocked edges incident to a common cluster, at most one of the clusters can be non-isolated (i.e. the cluster has multiple children).

\begin{lemma} \label{lem:push_down_group}
    For any set of $k$ level $i$ clusters in the same cluster graph whose combined size is $\leq \sizeconst{i}$, at most one of the clusters can have multiple children (it is not isolated).
\end{lemma}
\begin{proof}
    Assume that two or more of the clusters are not isolated. Then there are at least two disjoint blocked edges between level ${(i-1)}$ clusters. Each blocked edge must have $>\sizeconst{i-1}$ size. Having at least two of these, the total size is $>\sizeconst{i}$ which is a contradiction.
\end{proof}

This means that all but one of the cluster graphs will consist of a single cluster. 
Given this observation, we avoid the challenge of merging several level $i$ local trees in parallel, and simply need to \revision{delete the isolated clusters from the level $i$ cluster graph and insert them into the cluster graph of the level $(i-1)$ cluster that is not isolated, both of which can be done batch-parallel using our batch-dynamic local trees. We include the full description and analysis} for \textbf{PushDownGroup} in the appendix.
\begin{lemma} \label{common_push_down}
    Pushing down a batch of $k$ level $i$ edges incident to a common level $(i-1)$ cluster can be done in $O(\log n)$ depth and $O(k\log(1+n/k))$ work.
\end{lemma}

%We describe the 

Given this routine, we can reduce the problem of pushing down an arbitrary batch of edges to multiple calls of \textbf{PushDownGroup}.
The general strategy is to take a spanning forest of the edges and decompose it into disjoint stars. Then for each star we merge into the center a maximal prefix of the clusters sorted by size that can be merged into the center without violating the size constraint. We give the algorithmic details of \textbf{PushDownBatch} in the appendix, which proves the following lemma:

\begin{lemma} \label{lem:batch_push_down}
    Given a set $E$ of $k$ level $i$ edges, enforcing that every level $(i-1)$ cluster containing an endpoint of any edge in $E$ is incident to a blocked edge or its parent is isolated, can be done in $O(\log^2n)$ depth and $O(k \log n)$ work.
\end{lemma}

\subsection{Parallelizing Deletion} \label{sec:pardel}
Here we describe how to parallelize a single update in the blocked cluster forest using the results of the previous subsections. We focus on deletion since insertion already takes $O(\log^2n)$ depth.

\myparagraph{Pushing Down Edges}
During a deletion, the algorithm sweeps up the levels of the cluster forest until a replacement edge is found or the top is reached.
At each level some edges may be pushed down. In the sequential case, these edges were immediately pushed down as far as possible.
To parallelize deletion, we will collect all of the edges that are pushed down at each level during this upward sweep, and handle them later in a downward sweep to restore the blocked invariant. Let $E_\ell$ be the set of edges pushed down during the upward sweep of deletion at a level $\ell$.

Once the upward sweep has finished at level $\ell_u$, we start a downward sweep, starting at level $\ell_d = (\ell_u-1)$. At each level $\ell_d$ in the downward sweep, we call \textbf{BatchPushDown}$(E_{\ell_d})$. Every edge that is pushed down by this call is added to the set $E_{(\ell_d-1)}$.

\myparagraph{Restoring the Blocked Invariant}
When a cluster is split at higher levels during deletion, the algorithm must restore the blocked invariant in that cluster graph before continuing.
As in the sequential setting, if the split cluster was a satellite, we just need to find and/or push down a single outbound edge incident to both fragment clusters. 
Lemma~\ref{lem:doubling_fetch} proves that we can find such an edge in $O(\log^2 n)$ depth and $O(k \log n)$ work which can be charged to the self-loops that are found and pushed down.
If the split cluster was the center, restoring the blocked invariant is much more complex. Due to space constraints we leave the description of this to the appendix. Our result is summarized by the following lemma:

\begin{lemma}\label{lem:par_restore_inv}
    When a cluster $X$ is split in a cluster graph $CG(c)$ that previously maintained Invariant~\ref{inv:blocked}, the invariant can be restored in $O(\log^2 n)$ depth and $O((k+1) \log n)$ work where $k$ is the total number of edges pushed down by this process.
\end{lemma}

We leave the full analysis of deletion to the appendix. Our result is the following lemma:

\begin{lemma}\label{lem:par_del}
    A deletion in a blocked cluster forest can be done in $O(\log^3 n)$ depth and $O(k\log n + \log^2 n)$ work where $k$ is the total number of edges pushed down during the deletion.
\end{lemma}
\section{Parallel Batch-Dynamic Updates}\label{sec:batchpar}
In this section we show how to extend blocked cluster forests to support parallel batch-dynamic operations.
%Our goal is to obtain good guarantees on the amortized work of each operation and show that each batch operation runs in provably low depth.
%Sections~\ref{sec:batchins} and~\ref{sec:batchdel} describe batch edge insertion and deletion. 
%Due to space constraints, we describe how the blocked invariant is restored when clusters are split during batch deletion in the full paper~\cite{fullpaper}.
%Section~\ref{sec:batchrestore} describes how 

\subsection{Batch Insertion}\label{sec:batchins}
Consider a batch $E$ of $k$ edge insertions.
All of the edges in the batch are inserted at level $(\lmax+1)$. Let the set of edges be $E_{\lmax+1}$.
Then we call \textbf{BatchPushDown} on $E_{\lmax+1}$. Let $E_{\lmax}$ be the set of edges that were pushed down by this call. We repeatedly call \textbf{BatchPushDown} on $E_{i}$, the set of edges pushed down by the call on $E_{i+1}$.
Algorithm~\ref{alg:batchins} shows the pseudo-code for batch insertion.

\begin{algorithm}
\caption{$\mathsf{BatchInsertion}(CF, E)$}
\label{alg:batchins}
\begin{algorithmic}[1]
    \State Insert all of $E$ into the edge lists of their endpoints
    \State Batch update bitmaps for level $(\lmax+1)$ edges
    \While{$|E| > 0$}
        \State $E \gets \mathsf{BatchPushDown}(E)$
    \EndWhile
\end{algorithmic}
\end{algorithm}

When a batch of edges is introduced into a level, only the clusters incident to those edges may violate the blocked invariant. Calling \textbf{BatchPushDown} ensures that these clusters follow the blocked invariant when it is finished. Doing this for every level ensures that the blocked invariant is maintained throughout the cluster forest.
Each call to \textbf{BatchPushDown} takes $O(\log^2 n)$ depth, so the total depth of batch insertion across all levels is $O(\log^3 n)$.
The work on the first call to \textbf{BatchPushDown} is $O(k \log n)$. This work can be charged to the $k$ edges in $E$ being inserted. Each subsequent level $i$ performs $O(|E_i| \log n)$ work which can be charged to the $|E_i|$ edges pushed down at the previous level.

\subsection{Batch Deletion}\label{sec:batchdel}
Consider a batch of $k$ edge deletions. The general strategy will be to keep an \defn{active set} $A$ of fragment clusters whose connectivity may have changed.
First there will be a bottom-up sweep on the levels at the cluster forest. At every level the algorithm (1) restores the blocked invariant (2) performs connectivity search, and (3) possibly splits clusters in the level above. In this upward sweep edges will be pushed down only one level which may temporarily violate the blocked invariant in the lower levels that have already been processed.
% The size of the active set will remain $O(k)$ because there can be at most $k+1$ components caused by $k$ edge deletions.
Afterwards there will be a downward sweep to push down any edges that were not pushed down as far as possible to maintain the blocked invariant.
Algorithm~\ref{alg:batchdel} shows the pseudo-code for batch deletion.

\begin{algorithm}
\caption{$\mathsf{BatchDeletion}(CF, E)$}
\label{alg:batchdel}
\begin{algorithmic}[1]
\revision{
\State Delete all of $E$ from the edge lists of their endpoints
\State Batch update bitmaps for each edge's level prior to deletion
\State $[E_1 \hdots E_{\lmax}] \gets \mathsf{semisort}(E)$, $\ell \gets 1$, $A \gets E_1$
\While{$|A| > 0$} \Comment{sweep up}
    \State $Groups \gets$ sort A by parent \label{line:sort}
    \For{$(Group,P) \in Groups$} \Comment{parallel for}
        \State $\mathsf{RestoreBlockedInvariant}(CG(P))$ \label{line:fix_invariant}
        \State $Components \gets \{\}$ \label{line:cc_start}
        \For{$Cluster \in Group$} \Comment{parallel for}
            \State $e \gets Cluster.\mathsf{FetchOutboundEdge}()$
            \If{$\neg e$} $Components.\mathsf{Insert}(Cluster)$
            \Else ~$Components.\mathsf{Insert}(null)$ \EndIf
        \EndFor \label{line:cc_end}
        \If{$|Components| > 1$} \label{line:one_component}
            \For{$(CC \neq null) \in Components$} \Comment{parallel for} \label{line:new_parents_start}
                \State Create a new parent cluster for $CC$
            \EndFor \label{line:new_parents_end}
            \State $P.\mathsf{RemoveChildren}(Components)$
            \State $\mathsf{Parent}(P).\mathsf{AddChildren}(\text{new parents})$ \label{line:add_children}
        \EndIf
    \EndFor
    \State $A \gets P~\cup$ new parents created \label{line:next_level_start}
    \State $\ell \gets \ell+1$, $A \gets A \cup \{ a \mid a \in e \in E_\ell \}$ \label{line:next_level_end}
\EndWhile
\For{$\ell \in [\lmax-1, 1]$} \label{line:sweep_down_start} \Comment{sweep down}
    \State $D_{(\ell-1)} \gets D_{(\ell-1)} \cup \mathsf{BatchPushDown}(D_\ell)$
\EndFor \label{line:sweep_down_end}
}
\end{algorithmic}
\end{algorithm}

First all of the edges in the batch are deleted from the list of edges at the leaves of the cluster forest and the bitmaps are updated. The next step is to semisort all of the deleted edges in the batch by their (former) level.
During the upward sweep at level $i$, we will add the level $i$ clusters containing the endpoint of each level $(i+1)$ edge to our set of fragment clusters.
When subject to $k$ deletions, clusters at upper levels may be split into as many as $k+1$ fragments now instead of just two. 
%Figure~\ref{fig:multi_split} shows an example of a cluster graph before and after multiple splits to its center cluster.
Section~\ref{sec:batchrestore} describes how the blocked invariant can be restored in parallel given the possibility of several clusters being split by the batch of deletions. This yields Lemma~\ref{lem:batch_restore_invariant}.

% \begin{figure}
% \centering
% \includegraphics[width=0.35\textwidth]{figures/multi_split.png}
% \caption{\label{fig:multi_split}
% \small
%     An example of a cluster graph before and after the center cluster $P$ is split into three fragments.
% }
% \end{figure}

\begin{restatable}{lemma}{batchrestore} \label{lem:batch_restore_invariant}
    Given a cluster graph $CG(c)$ that was subject to $k$ clusters being split and  previously maintained Invariant~\ref{inv:blocked}, the invariant can be restored in $O(\log^2 n)$ depth and $O((x+k)\log n)$ work where $x$ is the total number of edges pushed down by this process.
\end{restatable}

\myparagraph{Parallel Connectivity Search and Splitting Clusters}
Consider the active set of $O(k)$ clusters at level $i$. They can be grouped based on their level $(i+1)$ parent cluster, by finding the parent of every cluster in $O(k \log n)$ work, and sorting the clusters by parent in $O(\log k)$ depth and $O(k \log k)$ work (line~\ref{line:sort})~\cite{cole1986sort}. Then each group can be processed independently in parallel.
Let $k'$ be the number of active clusters in a given group.
Within each group, we first fix the blocked invariant in the cluster graph of the parent (line~\ref{line:fix_invariant}). This takes $O(\log^2 n)$ depth and results in $O(k' \log n)$ uncharged work by Lemma~\ref{lem:batch_restore_invariant}.

Then we want to determine the connected components of the cluster graph $CG(P)$ of the parent $P$. We will take advantage of the fact there cannot be a matching of size greater than one over the blocked edges in $CG(P)$. Since the blocked invariant has been restored, only one connected component can contain multiple clusters in $CG(P)$.
The strategy will be to attempt to find an outbound edge for each cluster in the active set. This can be done in $O(\log^2 n)$ depth with $O(\log n)$ uncharged work per cluster, the rest of the work is charged to the self-loops that were found and pushed down. There will be $O(k' \log n)$ total uncharged work.
If an outbound edge is found, it is part of the component with multiple clusters. Each other component is defined by a single cluster. We produce a set of the connected components, representing each singleton component with its cluster and representing the component with multiple clusters as null if it exists (lines~\ref{line:cc_start}--\ref{line:cc_end}).

If there is only one component, the deletion is done within this component, meaning every edge deletion in this cluster graph has certified connectivity (line~\ref{line:one_component}). Otherwise, each lone cluster is removed as a child of $P$, and will get a new parent node at level $(i+1)$ (lines~\ref{line:new_parents_start}--\ref{line:new_parents_end}). These will be added to the modified local tree of the parent of $P$ using \textbf{BatchInsert}, and their sizes will be subtracted from $n(P)$ within the \textbf{AddChildren} function (line~\ref{line:add_children}). The component with multiple clusters will keep the original $P$ as its parent. If it didn't exist (e.g. $n(P)=0$ now), we delete $P$ by removing it as a child of its parent.

The active set at the next level up will be the set of new level $(i+1)$ clusters and possibly $P$ for any group that still had multiple components, along with the clusters containing endpoints of level $(i+2)$ edge deletions (lines~\ref{line:next_level_start}--\ref{line:next_level_end}).

\myparagraph{Downward Sweep}
Just like in parallelizing a single deletion, we will collect all of the edges that are pushed down at each level during the upward sweep, and handle them later in a downward sweep to restore the blocked invariant. Let $D_\ell$ be the set of edges pushed down during the upward sweep of deletion at a level $\ell$.
We start the downward sweep at level $\lmax-1$. At each level $\ell$ in the downward sweep, we call \textbf{BatchPushDown}$(D_{\ell})$ (lines~\ref{line:sweep_down_start}--\ref{line:sweep_down_end}). For every edge that this call pushes down, we add it to the set $D_{(\ell-1)}$.

\myparagraph{Cost Analysis}
\revision{We carefully analyze the work and depth of our batch update algorithms in the appendix. The result is summarized by the following theorem:}

\begin{restatable}{theorem}{batchupdate}
\label{thm:batch_par_updates}
    Batch insertions and batch deletions of edges in the blocked cluster forest can be done in $O(\log^3 n)$ depth per batch with an amortized work of $O(\log^2 n)$ per edge.
\end{restatable}

% \quinten{Outline of experiments (? indicates something that may be a stretch goal for 6/1 deadline):}
% \begin{enumerate}
%     \item Space Usage comparison between:
%     \begin{itemize}
%         \item HDT implementation
%         \item CWN implementation (blocked, unblocked)
%         \item CWN with path compression (blocked, unblocked)
%         \item Array implementation ?
%     \end{itemize}
%     \item Sequential update speed comparison:
%     \begin{itemize}
%         \item HDT implementation
%         \item CWN compressed vs uncompressed
%         \item CWN insert at root/LCA
%         \item CWN insert with blocked invariant
%         \item CWN delete with blocked invariant ?
%     \end{itemize}
%     \item Query speed comparison:
%     \begin{itemize}
%         \item HDT implementation
%         \item CWN uncompressed vs compressed
%     \end{itemize}
%     \item Cluster graph statistics on real-world graphs (average and max for cluster graph size and diameter):
%     \begin{itemize}
%         \item CWN insert at root/LCA
%         \item CWN insert with blocked invariant
%         \item CWN delete with blocked invariant ?
%     \end{itemize}
%     \item Scalable parallel implementations ?
% \end{enumerate}

% \quinten{TODOs}
% \begin{enumerate}
%     \item Implement blocked invariant for sequential deletions
%     \item Compare against D-Tree (space, update speed, query speed)
%     \item Run all experiments on more graphs (e.g. Twitter)
%     \item Collect diameters of cluster graphs on real-world data
%     \item Query performance
% \end{enumerate}

\section{Empirical evaluation}\label{sec:exp}
\revision{
In this section, we provide the first experimental study of the cluster forest algorithm.
Our goals in this part of the paper are to understand (1) whether the cluster forest algorithm is practical and yields good query and update performance; (2) whether the theoretical space improvements provided by the cluster forest algorithm translate into meaningful space improvements in practice; and (3) whether the algorithm can scale to large graphs and work well across a variety of different graph types (both real-world and synthetic).
Since there have been no prior implementations of the cluster forest approach, we focus our study on sequential implementations in order to carefully study different design choices in the algorithm and carefully measure the impact of these choices on runtime and query performance.
In this section, we demonstrate that a carefully optimized implementation of the cluster forest algorithm can achieve all three of these goals, and that the cluster forest approach may be the algorithm of choice when both theoretical guarantees and practical performance are required.
}
%\quinten{Beginning: Summarize the main findings / takeaways of the experimental section.}
%
%(Prelim): introduce other competing implementations, HDT, D-Tree.
%
%(1) Space usage
%- plots showing space usage vs batch
%- discussion / explanation of the results (uncompressed vs compressed, etc).
%- discussion of 2*nlog(n) vs 4*nlog(n)
%
%Paragraph: Space vs existing approaches:
%- discuss HDT vs D-Tree vs CF
%
%Paragraph: takeaways
%
%(2) Update performance
%- plots showing update performance vs batch
%- discussion / explanation of the results (uncompressed vs compressed, etc).
%- explaining why:
%    - total number of edges fetched
%    - path length stuff
%    - maybe include plots for the above explanatory factors 
%
%Paragraph with takeaways.
%
%(3) Query performance
%- show that CF does not sacrifice query performance.
%
%(4) Cluster graph properties
%- avg cluster graph size (above some level?)
%- avg diameter (above some level?) (time permitting) 
%- histograms instead of avg might be useful.

\revision{\subsection{CF Algorithm Optimization}\label{sec:implementation}} 

\revision{
% Bottleneck of Cluster Forest updates is traversing up from leaves, all of our implementation effort is to reduce this time:
% For example: what we did before is when merging a group of clusters just combine the clusters instead of checking for each edge if they should combine clusters
% Keep track of non-tree edges for fast deletions
% 2-level queue to efficiently traverse all level i edges
% Root vs LCA tradeoff for insertion/deletion speeds

One of the main contributions of this paper is the first practical and highly-optimized implementation of the cluster forest algorithm. 
A major bottleneck when implementing the algorithm is the cost of traversing the cluster-forest hierarchy, a step that occurs in nearly all aspects of the algorithm (e.g., replacement edge search, fetching a level $i$ edge, or pushing an edge from level $i$ to $i-1$).
Although the hierarchy has depth $O(\log n)$, traversals can still be costly as they encounter both cluster forest nodes and local tree nodes, and thus every traversal involves significant pointer jumping.
To address this issue, we designed our implementation to reduce the cost of and eliminate tree traversals whenever possible.
Due to space constraints, we provide more a detailed description and discussion of our optimizations in the appendix.

Our first optimization is called \defn{flattened local trees}. As the name suggests,
we flatten the local tree structure into an array that store the roots of the rank trees in increasing order of rank instead of combining the rank trees into a binary tree (see Section~\ref{sec:seqcf} for local and rank tree definitions). 
Since there are at most $\log_2 n$ rank trees, this array approach does not sacrifice the time complexity of any local tree operation and improves locality.
Additionally, we do not combine rank trees nodes unless there are more than $\log_2 n$ of them. This means that in many cases where a node in the cluster forest has few children, the algorithm can avoid a large amount of indirection in traversing the local trees and rank trees.

Our next optimization is called \defn{lowest common ancestor (LCA) insertion}.
The optimization is to insert an edge at the level of the lowest level node in the cluster forest that contains both endpoints of the edge.
A faithful implementation of the cluster forest algorithm is to simply perform \defn{root insertion}, i.e., every non-tree edge insertion is simply placed at the root of its tree.
%With root insertion, if we delete an edge with a replacement edge we find it quickly. 
%But with no replacement we may traverse all of the edges near the top.
The LCA optimization trades off extra time spent during an insertion to find the LCA to distribute the edges better across the levels to achieve faster deletions, since fewer edges need to be searched, thus lowering the amount of traversals.
The performance overhead of performing LCA insertion is negligible compared to root insertion---on average LCA insertion is $1.06\times$ faster on the graphs we evaluated when comparing total insertion time.
However, LCA insertion makes a huge difference for speeding up deletions---on average, it speeds up total deletion time by $1.5\times$ on average over root insertion across all of our graphs.
%Compared to root insertion, across all of our graphs, LCA insertions slow down insertion by only $X\%$ on average and speed up deletion by $Y\%$ on average.\laxman{TODO: don't forget numbers here.}
%
We note that the LCA optimization is somewhat unique to the cluster forest algorithm, since finding the LCA can be done in $O(\log n)$ time using the cluster forest representation; applying a similar optimization to HDT seems more complex as it requires quickly finding the lowest level where the endpoints of an edge are connected, which naively takes $O(\log^2 n)$ time.

%We observed that this effect was particularly strong on sparse graphs where more deletions affect the connectivity of the graph.

When evaluating an early version of our algorithm, we found that on dense graphs, our implementation was slower than HDT, even when using the LCA optimization.
The reason for HDTs speed is that $\geq m-n+1$ edges are {\em non-tree edges}, and therefore, many deletions in dense graphs target non-tree edges. 
HDT benefits from this fact, since a non-tree edge deletion simply checks a hash-table storing whether an edge is tree or non-tree in $O(1)$ time, and updates bitmaps after deleting the edge.
To achieve similar benefits in the cluster forest algorithm, we introduce a {\bf non-tree edge tracking} optimization.
Due to space constraints we give a full description of the optimization in the appendix.
In a nutshell, the idea of the optimization is to carefully mark edges as either tree or non-tree---as in HDT, non-tree edges can simply be deleted without affecting the connectivity information.
%
%However, unlike HDT more than $n-1$ edges can be marked as tree edges in our implementation since there is no natural notion of a tree edge in the CF algorithm.
%
%The benefit of the optimization is that non-tree edge deletions can be handled very similarly to HDT by only updating the bitmaps.
%
On dense graphs, the non-tree edge tracking optimization yields a significant speedup---for example, on Orkut, we observed a speedup of $1.87\times$ after implementing our non-tree edge tracking optimization due to 87\% of the deletions being detected as non-tree edges. 
Compared to HDT, which detects 89\% of the deletions as non-tree edges our optimization shows that our CF implementation can almost completely match the HDT implementation's ability to avoid unnecessary work during edge deletions.
}

\subsection{Experimental Setup}
All of the experiments presented in this paper were run on a machine with 4 $\times$ 2.1 GHz Intel Xeon(R) Platinum 8160 CPUs (each with 33MiB L3 cache) and 1.5TB of main memory.

\revision{
\myparagraph{Implementations}
Our cluster forest implementations are all written in C++ and use B-tree sets and flat hash sets from Abseil~\cite{abseil}.
We refer to our implementations with and without the LCA insertion optimization as \defn{CF-LCA} and \defn{CF-Root}, respectively.
We compare against a faithful implementation of the original HDT algorithm~\cite{holm2001poly} written in C++ also optimized using set data structures from Abseil. 
Our implementations all use -O3 optimization.
We also compare against D-Tree~\cite{chen2022spanning}, a recently published data structure for dynamic connectivity that is written in Python.
We note that since D-Tree is implemented in Python comparing its running time and memory usage with other implementations written in C++ may be unfair; however, since D-tree is a recent linear-space dynamic connectivity algorithm, we include it for completeness.
}

\myparagraph{Input Data}
The graphs used in our experiments are summarized in Table~\ref{tab:graph_data}. 
\revision{We use a variety of real-world and synthetic graphs with varying size, density, and type.}
To generate dynamic updates, we generate a random permutation of all of the edges in the graph, \revision{and insert all of the edges in this order.} Then we generate another random permutation of the edges \revision{and delete all of the edges in this order.
We break these random permutations into 10 \defn{stages} of inserting $|E|/10$ edges per-stage and then 10 stages of deleting $|E|/10$ edges. 
We use stages to understand how memory usage, and the insert and delete speeds change over time (e.g., as the graph grows more dense, and then more sparse).
%measure memory usage periodically, to get a more fine-grained breakdown of update speeds, and to insert queries into the sequence of updates. 
At the end of each stage we perform 1M queries: half the queries are completely random and half are endpoints of edges that exist in the graph at that point.}

\begin{table}[]
\revision{
\footnotesize{}
    \centering
    \vspace{-2em}
    \caption{\small Graph datasets used in our experiments.}
    \begin{tabular}{|l|l|l l l|l|}
    \hline
        Name & Type & $|V|$ & $|E|$ & Avg. Deg. & Cite \\
        \hline \hline
        Germany Roads (GER) & Road & 12.28M & 16.12M & 2.62 & \cite{roadgraphs} \\
        USA Roads (USA) & Road & 23.95M & 28.85M & 2.41 & \cite{roadgraphs} \\
        Household Lines (HH) & $K$-NN & 2.05M & 6.50M & 6.35 & \cite{ucimlrepo} \\
        Chem (CHEM) & $K$-NN & 4.21M & 14.83M & 7.05 & \cite{chem5} \\
        Youtube (YT) & Web & 1.16M & 2.99M & 5.16 & \cite{yang2012defining} \\
        Pokec (POKE) & Web & 1.63M & 22.30M & 27.32 & \cite{snapnets} \\
        Wiki Topcats (WT) & Web & 1.79M & 25.44M & 28.41 & \cite{yin2017local} \\
        ENWiki (EW) & Web & 4.21M & 91.94M & 43.72 & \cite{boldi2004webgraph} \\
        Skitter (SKIT) & AS & 1.70M & 11.10M & 13.08 & \cite{leskovec2005graphs} \\
        StackOverflow (SO) & Temporal & 6.02M & 28.18M & 9.36 & \cite{paranjape2017motifs} \\
        LiveJournal (LJ) & Social & 4.85M & 42.85M & 17.68 & \cite{backstrom2006group} \\
        Orkut (ORK) & Social & 3.07M & 117.19M & 76.28 & \cite{yang2012defining} \\
        Twitter (TWIT) & Social & 41.65M & 1.20B & 57.74 & \cite{kwak2010what} \\
        Friendster (FR) & Social & 65.61M & 1.81B & 55.06 & \cite{yang2012defining} \\
        Grid (GRID) & Synthetic & 10.00M & 10.22M & 2.04 & \cite{lattices} \\
        RMAT-26 (RMAT) & Synthetic & 67.11M & 670.83M & 19.99 & \cite{rmatcite} \\
        \hline
    \end{tabular}
    \label{tab:graph_data}
    }
\end{table}

\subsection{Performance Results}

\begin{figure*}[ht]
\vspace{-2em}
    \subfloat{\includegraphics[width=\textwidth]{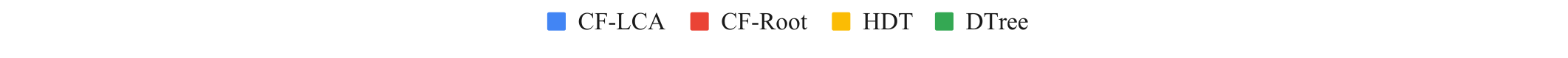}} \\
    \setcounter{subfigure}{0}
    \vspace{-1em}
    \subfloat{\includegraphics[width=\textwidth]{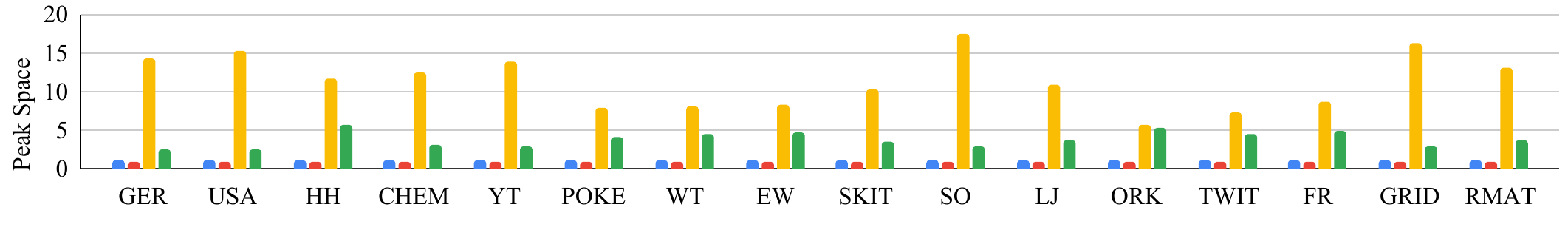}} \\
    \vspace{-0.5em}
    \subfloat{\includegraphics[width=\textwidth]{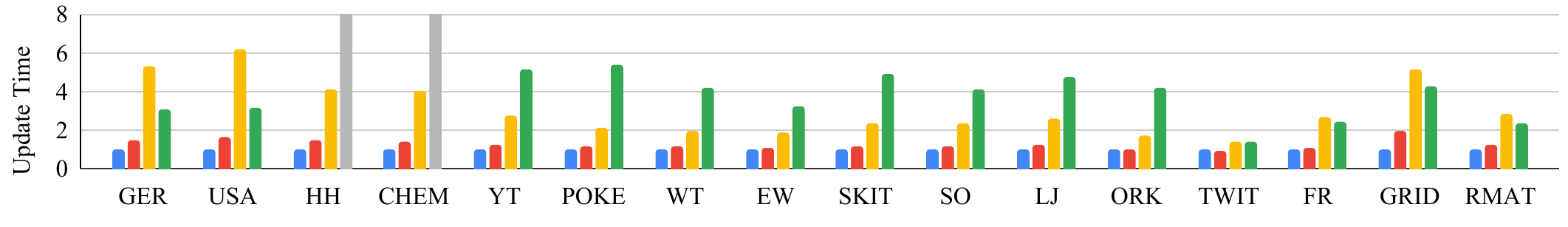}} \\
    \vspace{-0.5em}
    \subfloat{\includegraphics[width=\textwidth]{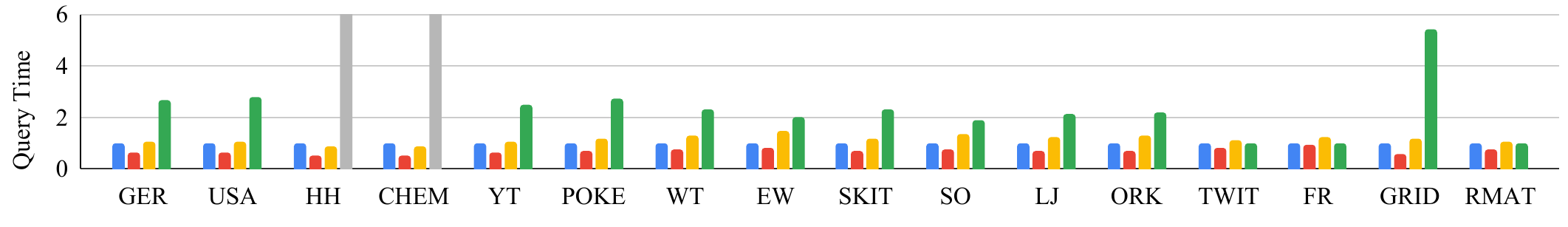}}
    \vspace{-0.3em}
    \caption{\revision{\small The results of our experiments for peak space usage (top), total update time (middle), and total query time (bottom) of each system on various inputs. All values are normalized to CF-LCA. The un-normalized results are presented in the appendix. A gray bar indicates that the time for D-Tree was over $100\times$ longer than that of CF-LCA on the same input or it terminated before completion with a timeout of 24 hours.}}
    \vspace{-1em}
    \label{fig:bar_plot_results}
\end{figure*}

\myparagraph{Memory Usage}
We start by investigating whether the theoretical guarantees on space provided by the CF algorithm translate into practical improvements in memory usage.
Our goal is to determine (1) whether memory usage is a limiting factor in the ability of HDT-based implementations to scale to extremely large graphs, (2) whether implementations based on the CF algorithm using linear space can overcome this obstacle, and (3) whether the memory usage of implementations based on the CF algorithm can perform better than existing state-of-the-art dynamic connectivity implementations using linear space.
Our results affirmatively answer all three of these questions.

\revision{The top of Figure~\ref{fig:bar_plot_results} shows the peak memory usage of each algorithm on the various graphs relative to the memory of CF-LCA. 
We report the raw (unnormalized) numbers for peak memory usage in the appendix.
}
Our main finding is that our CF implementations consistently require significantly less memory than HDT.
{\bf \emph{\revision{In particular, CF-Root uses 6.2$\times$--19.7$\times$ less memory than HDT and CF-LCA uses 5.7$\times$--17.5$\times$ less memory than HDT.}}}
%{\bf \emph{\revision{The CF-Root algorithm} uses \revision{5.5--16.1\%} of the \revision{peak} memory usage of our HDT implementation.}}
% {\color{green} I'd suggest saying that we obtain a $6.2-18.2x$ reduction instead}
%
%
%On average \revision{$9.3\times$} of the memory that HDT uses \revision{for the graphs that we tested.}
\revision{We note that across all the graphs we tested, the memory efficiency of both CF-Root and CF-LCA are very similar, with CF-Root having a slight edge ($1.0\times$--$1.2\times$ less memory). This is because with root insertion there are slightly fewer nodes in the cluster forest.}
\revision{On the densest graph we tested (Orkut) our best CF implementation uses 26 bytes per edge while HDT uses 159 bytes per edge. 
On the sparsest graph we tested (Grid) the CF algorithm uses 294 bytes per edge while HDT uses 5,809 bytes per edge.}

\revision{These results clearly show the benefits of the CF algorithm over HDT in practice across a wide variety of real-world graphs with different characteristics.
For sparse graphs (e.g. USA Roads, StackOverflow, Grid) where $m$ is closer to $n$, the $O(n+m)$ space usage of CF should beat the $O(n\log n + m)$ space usage of HDT, and this asymptotic difference can be clearly seen in the results.
Interestingly, even for dense graphs (e.g. ENWiki, Orkut, Twitter) where $m > n \log_2 n$, CF still uses significantly less memory than HDT.
This is because in dynamic connectivity algorithms edges can be stored very space efficiently ($\approx 10$ bytes per edge) while the space overhead of the tree data structures used in virtually all dynamic connectivity algorithms scales heavily with the number of vertices (at least a hundred bytes per vertex per tree).
This supports the conclusion that on real-world graphs (which are typically quite sparse) dynamic connectivity algorithms that require storing a large number of trees over the vertices (like HDT) prevent scaling to very large graphs.
The extra $\log n$ factor on the space usage, which may easily be overlooked in theory, {\em has a massive impact in practice}.}

\revision{
For the third question, we compare our implementations of the CF approach with an existing implementation of D-Tree~\cite{chen2022spanning}, which uses linear space in theory but sacrifices worst-case theoretical guarantees on update time.
Although it is implemented in Python, D-Tree still uses significantly less memory than HDT in all of our experiments. 
This is again indicative of the large impact of using data structures that have linear total space in practice.}
%\quinten{Is there anything else we want to say about D-Tree memory usage?}

%{\bf \emph{\revision{Our implementations of the CF algorithm use 29\%-32\%} of the memory used by the D-Tree implementation.}}

\begin{figure}
    \centering
    \vspace{-1em}
    \includegraphics[width=\linewidth]{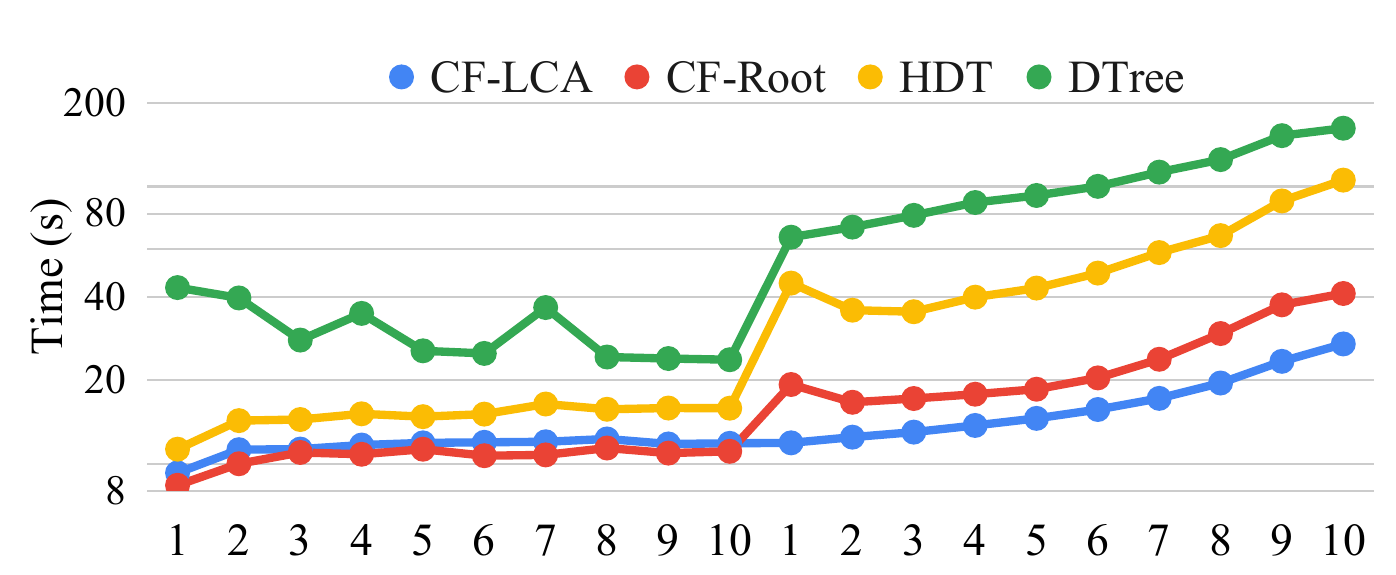}
    \caption{\revision{\small The total time in seconds for each stage of updates on the LiveJournal graph. The first 1--10 on the $x$-axis represent the 10 stages of insertion and the next 1--10 represent the 10 stages of deletion.}}
    \label{fig:stage_plot}
\end{figure}

\myparagraph{Update Times}
The goals of our next experiments are (1) to determine whether CF implementations can match or beat the performance of HDT for updates despite the increased implementation complexity of CF, (2) to compare the update performance of our CF implementations with the performance of existing state-of-the-art dynamic connectivity implementations, \revision{and (3) to investigate the impact of the optimizations from Section~\ref{sec:implementation} on update speed.}

\revision{The middle of Figure~\ref{fig:bar_plot_results} shows total update time of each algorithm on the various graphs relative to the time for CF-LCA. We report the raw (unnormalized) numbers for updates times in the appendix. Additionally, we include Figure~\ref{fig:stage_plot} to show an example of how the update performance varies throughout the sequence of updates (per-stage) on the LiveJournal graph.}

Our results show that our CF implementations can achieve \revision{significantly better performance than HDT for both insertions and deletions for all all graphs in our experiments.
{\bf \emph{CF-LCA performs updates 1.4$\times$--6.2$\times$ faster than HDT, and CF-Root performs updates 1.5$\times$--3.8$\times$ faster than HDT.}}
The improvements in update speed for the CF implementations can be attributed to the optimizations we described in Section~\ref{sec:implementation}.
For example, for non-tree edge insertions, both algorithms must simply add the edge to the edge set, and then traverse up to update $O(\log n)$ bitmaps; the flattened local tree optimization decreases the cost of such traversals in the CF implementation resulting in faster insertions.
For non-tree edge deletions, the non-tree edge tracking optimization enables our CF implementation to match the performance of HDT. 
With the optimization, both algorithms simply check the hash-table once, delete the edges from the edge set, and traverse up to update the bitmaps.
The flattened local tree optimization once again allows the CF implementations to outperform the HDT implementation due to the lower cost of traversing the hierarchy.
}

\revision{
CF-LCA performs better than CF-Root for total update time in most cases.
This is because while the LCA optimization slows down insertions slightly, deletions become much faster due to fewer edges needing to be inspected.
However, CF-Root benefits from being able to find replacement edges very quickly in dense, well connected graphs like Twitter.
}

We also compared our CF implementations with D-Tree, which sacrifices worst-case guarantees but obtains linear space usage.
\revision{The update performance of D-Tree varies greatly depending on the graph due to its heuristic nature, but the CF algorithms perform updates faster than D-Tree in all of our experiments. We note that updates for D-Tree can take much longer on certain graphs; e.g., D-Tree is $173\times$ slower than CF-LCA on the Household Lines graph. Since the algorithm does not have worst-case guarantees, we believe it is unlikely to be much faster on these bad cases even if it was implemented in a different language.}

\revision{
\myparagraph{Impact of Graph Properties}
Interestingly, graph size does not seem to play a major role in the relative performance of different dynamic connectivity algorithms, and other properties such as graph density, and graph diameter play a more important role.
For example, we observe a strong positive correlation between density and improved running time for HDT; e.g., despite Twitter being one of the largest graphs we evaluate on, its high density make HDT perform significantly better.
%We note that our CF and HDT implementations are more robust to adversarial graph inputs compared to D-Tree---for example, on high-diameter graphs, D-Tree has essentially unbounded running time due to degrading to a worst case of $O(m+n)$ running time per-update, while CF and HDT continue to perform well on these graphs despite their high diameter.\laxman{TODO: update once Zhongqi's numbers are in.}
}

\myparagraph{Query Times}
\revision{The bottom of Figure~\ref{fig:bar_plot_results} shows the total query time of each algorithm on the various graphs normalized to the time for CF-LCA. We report the raw (unnormalized) numbers for query time in the appendix.}
Our main finding is that our CF implementations do not sacrifice any performance in terms of query speed compared to HDT.
This is expected, because in both types of algorithms, queries are answered by traversing to the root of the components of the two vertices which requires $O(\log n)$ memory reads and a single equality check.
\revision{Our results show that queries for our CF implementations are slightly faster than HDT.
We believe this is further indicative that the height of the cluster forest hierarchy is generally lower than that of the top level spanning forest data structure in HDT due to the flattened local tree optimization.}
\revision{The query speed for D-Tree varies greatly depending on the graph due to its heuristic nature, but is 
 almost always beaten by the CF algorithms and HDT in our experiments.}
%\quinten{D-Tree seems to match HDT and CF-LCA for very low diameter graphs (e.g. Twitter, Friendster, RMAT.}
\revision{Comparing CF-LCA and CF-Root, we find that CF-Root always has faster query times because the cluster forest hierarchy has smaller height when performing root insertions.}

%\myparagraph{LCA vs Root Insertion}

%\subsection{Discussion}
%
%\myparagraph{LCA vs Root}
%\quinten{general effect}
%\quinten{Impact of density}
%
%\myparagraph{Other findings}

%\subsection{Statistic Results}
%\quinten{Include results for cluster graph sizes and diameter which indicate the necessity of our theory}

\section{Conclusion}
%Prior to our work, all existing CF algorithms were sequential, and no implementation of a CF algorithm had been studied before. An open theory question, therefore, is whether one can adapt the CF approach to the parallel batch-dynamic setting.
%Existing practical work on dynamic connectivity focused on HDT, and simpler heuristic methods like D-Tree to achieve low space.
%

\revision{This paper makes two significant contributions towards developing a scalable and practical batch-dynamic connectivity algorithm.
First, on the theoretical side, we give the first parallel batch-dynamic algorithm for maintaining the connected components of an undirected graph that is work-efficient, runs in poly-logarithmic depth, and only uses linear total space.
Second, we give the first empirical study of the cluster forest algorithm in the sequential setting, introduce new optimizations to improve its practicality, and demonstrate its superior performance and space-efficiency in practice.
Taken together, our results indicate that the CF algorithm is an excellent candidate for a practically scalable dynamic connected components algorithm with good theoretical guarantees.
}

%It is interesting to note that despite several theoretical papers in recent years designing parallel (batch-)dynamic connectivity implementations (including this paper), there are no implementations of any parallel dynamic connectivity algorithms.

%We focused on sequential implementations of the CF algorithm since even in this setting, there was a large unexplored design  space for optimizations and implementation techniques that was interesting and non-trivial to .

%In future work, we plan to study {\em parallel} dynamic connectivity in practice by building on the theoretical and practical algorithms of this paper.
%An impo CF approach and studying their speedups is a very interesting future question.
%The theoretical algorithm we introduced in this paper is an interesting candidate since it is based on the blocked invariant, which we showed has significant practical potential.
%
%A major takeaway of our study is the importance of good guarantees in minimizing both update time and space usage, and we believe these lessons will extend to the parallel setting.
%potential of the CF approach in practice by providing the first implementation of the CF approach and validating its good space and update complexity on large real-world data.
%\laxman{think about this...}

%% Acknowledgements

% (include in the final paper)
%\section{Acknowledgements}
%We thank Daniel DeLayo (Stony Brook University) for his help in developing a tool to measure the memory usage of D-trees.

\balance
\bibliographystyle{ACM-Reference-Format}
\bibliography{references}

\clearpage{}

\begin{appendix}
\section{Appendix}

\subsection{Properties of Blocked Cluster Forest}

% \adam{Added the definition of the center node to the below lemma since it is used in Lemma~\ref{lem:heavy_center}.}

\blockedmatchingsize*
\begin{proof}
    Suppose $M > 1$. Then there are at least two blocked edges $(u,v)$ and $(w,y)$ on four vertices in $M$. Since both $(u,v)$ and $(w,y)$ are blocked, 
    $n(u) + n(v) > 2^{i-1}$, and similarly
    ${n(w) + n(y) > 2^{i-1}}$. 
    But this means that $2^i=2\cdot 2^{i-1} < n(u) + n(v) + n(w) + n(y) \leq n(c)$, which means that $c$ violates Invariant~\ref{inv:size} at level $i$.
\end{proof}

\diameter*
\begin{proof}
Suppose $CG(c)$ has diameter $\mathsf{diam}(CG(c)) > 2$, and let ${P = p_1, p_2, p_3, p_4 \ldots}$ be a path achieving the diameter, i.e., the shortest path between $p_1$ and $p_4$ uses the edges $\{(p_1, p_2), (p_2, p_3), (p_3, p_4)\}$.
Since $CG(c)$ satisfies Invariant~\ref{inv:blocked}, each of $p_1, p_2, p_3, p_4$ are incident to a blocked edge.

The neighbor along a blocked edge for $p_1$ and $p_4$ cannot be the same, since otherwise the shortest path between $p_1$ and $p_4$ would have length $2$, contradicting that the diameter is greater than $2$. Let $u, v$ be endpoints of the blocked edges incident to $p_1, p_4$ respectively.
Since $u \neq v$, $\{(p_1, u), (p_4, v)\}$ form a blocked matching of size $> 1$, which by Lemma~\ref{lem:blocked_matching_size} means that $CG(c)$ violates Invariant~\ref{inv:size}.
\end{proof}

\radius*
\begin{proof}
Consider adding the blocked edges to the graph incrementally.
The first edge could be anywhere. Call that edge $e_1=(u,v_1)$.
When adding a second edge it must be incident to the first edge. If it wasn't, the 4 endpoints of the two blocked edges together would violate the size invariant for $c$. Without loss of generality, say the second edge is incident to $u$, it is $e_2=(u,v_2)$.
When adding a third edge it will either be (1) $e_3=(v_1,v_2)$ forming a triangle or (2) $e_3=(u,v_3)$, an edge incident to the common node of the first two blocked edges.

In case (1) where we have a triangle, we cannot add anymore blocked edges without there being two blocked edges with four distinct vertices which would violate the size constraint for the level $i$ parent cluster. Therefore we are done and the lemma holds using any of the 3 vertices of the triangle as the center.

In case (2) where it does not form a triangle it must be incident to $u$. If it was not incident to any of $e_1$ or $e_2$, there would be 2 blocked edges with four distinct vertices which we cannot have. If it was incident to only $v_1$ or only $v_2$, then $e_3$ and $e_1$ or $e_2$ respectively would be form two blocked edges with four distinct vertices which we cannot have. Therefore $e_3$ must be incident to $u$. Any edges added after must be incident to $u$ also for the same reasoning. Therefore the lemma holds using $u$ as the center.
\end{proof}

\heavycenter*
\begin{proof}
Suppose that the center vertex of $CG(c)$ is not the largest cluster in $CG(c)$.
Let $u$ be the center vertex of $CG(c)$ and $S=\{v_1, v_2, \hdots, v_{k-1}\}$ be the $k-1$ remaining vertices. Lemma~\ref{lem:radius} tells us that $\forall v \in S$ there is a blocked edge $(u,v)$. By definition of a blocked edge this means that $n(u)+n(v) > \sizeconst{i-1}$, so $n(v) > \sizeconst{i-1} - n(u)$. Also note that $n(u) + \sum n(v) = n(c) \leq \sizeconst{i}$ by Invariant~\ref{inv:size}, and thus
$$n(u)+(k-1)(\sizeconst{i-1} - n(u)) < n(u) + \sum n(v) \leq \sizeconst{i},$$
that is
$$(2-k)n(u)+(k-1)\cdot \sizeconst{i-1} < \sizeconst{i}.$$
Consequently, since $k\geq 4$,
$$n(u) > \frac{(k-1)\cdot \sizeconst{i-1} - \sizeconst{i}}{(k-2)} = \frac{(k-3)\cdot \sizeconst{i-1}}{(k-2)}\geq \frac{\sizeconst{i-1}}{2}.$$

%\zhongqi{I think for the first a several lemmas, we are still using the old level representations and we need to update it. The final version we plan to use is that each level i cluster incidents to at most $2^i$ vertices in leaves.} \quinten{Updated to use new level notation.}

Let $n(u)=\sizeconst{i-1}/2 + c_0$ where $c_0 > 0$ is some constant. We will assume that $u$ is not the largest cluster and $\exists v \in S$ such that $n(v) \geq \sizeconst{i-1}/2 + c_0$. Then the remaining $k-2$ clusters have total size $\leq \sizeconst{i} - \sizeconst{i-1} - 2c_0 = \sizeconst{i-1} - 2c_0$. So at least one of these ($v^{*}$) must have size $\leq \sizeconst{i-1}/2 - c_0$. It follows that the blocked edge $(u,v^{*})$ could be pushed down since $n(u)+n(v^{*}) \leq \sizeconst{i-1}/2+c_0 + \sizeconst{i-1}/2-c_0 = \sizeconst{i-1}$. This is a contradiction.
\end{proof}

%\julian{where does the $-1$ come from?} \zhongqi{Another supplemental explanation here is, assume $k = 4$, then we have 2 remaining clusters with total size at most $\lfloor n/2^{i+1} \rfloor - 2c_0$. So if the two remaining clusters have equal size, then each of them should be $\lfloor n/2^{i+1} \rfloor/2 - c_0$. So anyone of these two other than cluster u can be pushed together with the larger cluster we mentioned above with size $\lfloor n/2^{i+1} \rfloor/2 + c_0$. This breaks the maximal matching lemma.} \quinten{True. Fixed to use Zhongqi's explanation.} 
%\adam{Maybe $v_k$ is a bit unfortunate pick here, since $k$ is a used variable name here.}

\unblockedisolated*
\begin{proof}
    Since $e$ is unblocked we know that $n(C_1)+n(C_2) \leq \sizeconst{i-1}$. Assume that $C_1$ is not an isolated cluster (if it is the proof is done). 
    We will show that $C_2$ must be isolated, thus completing the proof. 
    The cluster graph of $C_1$ contains at least one blocked edge between some level $i-2$ clusters $U$ and $V$. Since the edge is blocked, $n(C_1) \geq n(U)+n(V) > \sizeconst{i-2}$. Plugging this into the inequality from the beginning of this proof gives $n(C_2) < \sizeconst{i-2}$. Since $C_2$ does not have enough total size to have a blocked edge in its cluster graph, it must be isolated.
\end{proof}

\pushuntilblocked*
\begin{proof}
    % \adam{This is the first time we consider loops. Maybe add some indication that this can at all happen earlier?}\quinten{Added to definitions in prelims section}
    If the edge is a self-loop, it can be pushed down until it is no longer a self loop without changing the structure of the cluster forest \cf{}. 
    When an unblocked level $(i+1)$ edge between two distinct level $i$ clusters is pushed down, it is introduced as an new level $i$ edge between two previously disjoint cluster graphs. 
    The only clusters that can violate Invariant~\ref{inv:blocked} are the two level $(i-1)$ clusters containing each endpoint of the edge.
    Lemma~\ref{lem:unblocked_isolated} proves that one of the level $i$ clusters must be isolated (i.e., the cluster graph contains only a single cluster).
    If the edge is now blocked, the process ends and the invariant is maintained at this level.
    If they were both isolated clusters and the edge is still unblocked, it must be pushed down again, merging the two clusters into one and maintaining the invariant at this level.
    Else if they are not both single clusters and the edge is still unblocked, the single cluster is merged into the the other endpoint which must already be incident to a blocked edge; therefore, the invariant is maintained at this level.
\end{proof}

\deletioncorrect*
\begin{proof}
    Deletion only splits clusters or merges clusters that are connected by an unblocked edge, thus preserving the size invariant. Whenever an edge is pushed down during the deletion it is pushed down as far as possible, so by Lemma~\ref{lem:push_down_until_blocked} pushing down edges will not cause the blocked invariant to be violated. 
    The only other times a cluster can violate the the blocked invariant are (1) if it was incident to the deleted edge, (2) it was an ancestor of one of those clusters, or (3) its only blocked edge was to a center that lost mass from splitting. 
    Since the connectivity search algorithm proceeds until the search reaches a blocked edge or there are no more edges and it is a single cluster, cases (1) and (2) are prevented. The third case is prevented because the algorithm restores the blocked invariant whenever a cluster is split.
\end{proof}

\subsection{Useful Lemmas}
The following lemmas are useful for analyzing the work bounds of our
parallel algorithms. Proofs are provided by Acar et al.~\cite{acar2019parallel}.
\begin{restatable}{lemma}{componentbounds}\label{lem:component_bounds}
	Let $n_1, n_2, ..., n_c$ and $k_1, k_2, ..., k_c$ be sequences of non-negative integers such that $\sum k_i = k$, and $\sum n_i = n$. Then
	\begin{equation}
		\sum_{i=1}^c k_i \log\left(1 + \frac{n_i}{k_i}\right) \leq k \log\left(1 + \frac{n}{k}\right).
	\end{equation}
\end{restatable}

\begin{restatable}{lemma}{batchboundisrootdominated}\label{lem:batch_bound_is_root_dominated}
	For any non-negative integers $n$ and $r$,
	\begin{equation}
	\sum_{w=0}^{r} 2^{w} \log \left( 1 + \frac{n}{2^{w}} \right) = O\left( 2^{r}\log\left(1 + \frac{n}{2^{r}} \right) \right).
	\end{equation}
\end{restatable}

\begin{restatable}{lemma}{batchboundisincreasing}\label{lem:batch_bound_is_increasing}
	For any $n \geq 1$, the function $x \log\left(1 + \frac{n}{x}\right)$ is strictly increasing with respect to $x$ for $x \geq 1$.
\end{restatable}

\subsection{Modified Local Trees}
\label{sec:modified_local_trees}
We use a modified version of the local trees used in ~\cite{wulff2013faster} that also supports parallel operations. The interface is defined as follows:
\begin{itemize}
    \item \textbf{BatchInsert}$(c_1,...,c_k)$ takes a list of $k$ new leaf nodes and adds them to the local tree.
    \item \textbf{BatchDelete}$(c_1,...,c_k)$ takes a list of $k$ leaf nodes in $T$ and removes them from the local tree.
    \item \textbf{Smallest}$(C)$ returns the smallest element in the local tree.
    \item \textbf{GetMaximalPrefix}$(s)$ returns a maximal set of the elements in the local tree such that their total size is less than or equal to $s$ and no element excluded from the set has size less than any included element (e.g. the set returned is a prefix of the elements sorted by increasing size).
    \item \textbf{GetMaximalPrefixOfSubset}$(C,s)$ returns a maximal set of the elements in the subset $C$ of the elements in the local tree such that their total size is less than or equal to $s$ and no element excluded from the set has size less than any included element.
    \item \textbf{Mark}$(c)$ marks an element to indicate something about it.
    \item \textbf{GetUnmarkedElement}$()$ return an element in $T$ that has not been marked.
\end{itemize}

The tree is structured as follows. The elements are divided into size classes where two elements $x$ and $y$ are in the same size class if $\lfloor \log n(x) \rfloor = \lfloor \log n(y) \rfloor$. For each size class, a weight-balanced tree is maintained for all of the elements in that class. The weight of any leaf in the weight-balanced tree is just $1$. The key for each element in the weight-balanced tree is its size. The weight-balanced trees are also augmented with the total size of elements in a subtree. Then the usual rank trees and local tree are applied, treating each weight-balanced tree $t$ as a leaf with rank $\lfloor \log n(t) \rfloor$ where $n(t)$ is the total size of all the elements in $t$.

\begin{lemma}
    The depth of any leaf node in a cluster forest using modified local trees is $O(\log n)$.
\end{lemma}
\begin{proof}
    At every node in the weight-balanced trees, the total size in the left and right subtree is still off by at most a constant factor because every element in the same weight class has at most twice the size of any other element. Then for the local trees the same argument as before applies. For a child $v$ of a node $u$ in the cluster forest, the depth of $v$ in the local tree is at most $\log(n(u)/n(v)) + 1$. Since there are $\log n$ levels in the cluster forest, a telescoping sums argument implies that any leaf of the local tree forest has depth $O(\log n)$.
\end{proof}

An insertion or deletion into a modified local tree works as follows. First insert or the delete the element from the weight-balanced tree for its size class. Then the total size in its tree will have changed. Since each weight-balanced tree is treated as a leaf in the rank trees and local tree, just delete and re-insert the leaf for that tree with its new total size as its rank.

For batch insertion or deletion first semisort the batch by weight class. Then for each class, do a batch insertion or deletion into the weight-balanced tree for that class. Finally, the entire rank trees and local tree are removed and rebuilt from scratch.

\begin{lemma}
    Inserting or deleting a batch of $k$ elements into a modified local tree can be done in $O(\log n)$ depth and $O(k\log(1+n/k))$ work where $n$ is the total number of elements in tree.
\end{lemma}
\begin{proof}
    The semisort takes $O(k)$ work and $O(\log n)$ depth. Batch insertions and deletions into the weight-balanced tree for a size class $c$ can be done in $O(\log n_c)$ depth and $O(k_c\log(1+n_c/k_c))$ work where $n_c$ is the number of elements in the weight-balanced tree for size class $c$ and $k_c$ is the number of elements being inserted into this tree. By applying Lemma~\ref{lem:component_bounds} the total work here is $O(k\log(1+n/k))$ and the depth is $O(\log n)$. Then deleting and rebuilding the rank trees and local tree can be done in $O(\log n)$ time sequentially because there are at most $O(\log n)$ size classes. This is less than or equal to $O(k\log(1+n/k))$ work by Lemma~\ref{lem:batch_bound_is_increasing}.
\end{proof}

To get the smallest element in the tree one can simply find the weight-balanced tree for the minimum size class in $O(\log n)$ time, and then get the smallest key element in that tree, also in $O(\log n)$ time.

The $\mathsf{GetMaximalPrefix}$ procedure is implemented as follows. First an exhaustive traversal of the rank trees and local tree is to find the roots of the weight-balanced trees for every size class. Then they are sorted by the size class they represent. The total size augmentation can be used to get a prefix of the trees with total size $\leq s$. In each of those return all of the elements. Then in the next smallest size class use the total size augmentation to return a prefix of the elements in this tree with total size less than or equal to the remaining amount left.

\begin{lemma}\label{lem:wb_get_prefix}
    The $\mathsf{GetMaximalPrefix}$ operation on a modified local tree with $n$ elements can be done in $O(\log n)$ depth and $O(k\log(1+n/k))$ work where $k$ is the number of elements returned.
\end{lemma}
\begin{proof}
    There are at most $\log n$ size classes so the exhaustive traversal takes $O(\log n)$ depth and work. The sort can be done using bucket sort in $O(\log n)$ time and space (the total space is still $O(n)$ because at level $i$ there are actually at most $\log(2^i) = i$ size classes). Then for each weight-balanced tree the prefix (or all of its elements) can be found in $O(k_c \log (1+n_c/k_c))$ work and $O(\log n)$ depth where $n_c$ is the number of elements in that tree and $k_c$ is the number of elements returned by that tree. The total depth is $O(\log n)$. The work in all weight-balanced trees sums to $O(k \log (1+n/k))$ by Lemma ~\ref{lem:component_bounds} and the $O(\log n)$ work in exhaustive traversal and bucket sort is less than or equal to that by Lemma~\ref{lem:batch_bound_is_increasing}.
\end{proof}

For the $\mathsf{GetMaximalPrefixOfSubset}$ operation, we take a bottom up approach starting from all the leaf clusters in the subset $C$, and mark all of their ancestors. Then we can simply think of it as doing a normal $\mathsf{GetMaximalPrefix}$ operation on the local tree defined by only the marked nodes. Similarly this can be done in $O(\log n)$ depth and $O(k\log(1+n/k))$ work.

For the $\mathsf{Mark}$ and $\mathsf{GetUnmarkedElement}$ operations, this can implemented by augmenting the tree with a bit that indicates if there is an unmarked element in the subtree of a node.

\subsection{Fetching Edges in Parallel}

To fetch a single edge we can simply traverse down to a leaf cluster containing a level $i$ edge by following a path of nodes that contain a set bit the bitmap for level $i$ edges. 
To fetch $k$ edges in parallel, we will maintain a set (of size at most $k$) of distinct paths down to root clusters that contain a set bits for level $i$ edges. 
If fewer than $k$ such leaves exist in the subtree of this cluster, we will fetch all of them.
%We will go down one level at a time in the local trees.

The batch algorithm proceeds level-by-level.
At every level, each current path will check if it has one or two children with a set bit. 
The total number of nodes with a set bit on the next level down is computed by summing the one or two possible child paths for each current path. 
If the total is $\leq k$, all nodes are taken. 
Else, we have between $k$ and $2k$ nodes and we will simply fetch a single leaf from each of the nodes' subtrees.

%\laxman{All of the following details should be pushed to the appendix. Let's just summarize with a Lemma and give a high-level idea of the challenge and how we overcome it. A naive implementation of the above idea will require $O(\log^2 n)$ depth, but we can explain the key idea for getting it down to $O(\log n)$ depth.}
To compute the total nodes at the next level we can use an approximate parallel prefix sum. 
In the CRCW PRAM model this can be done in $O(\log^*k)$ depth with high probability or in $O(\log\log k)$ depth deterministically. Using a $2$-approximate sum, we can say that if the approximate sum is $\geq 2k$, then the true sum is $\geq k$. this ensures that we will have between $k$ and $4k$ paths once we stop adding them. 
%\laxman{Mention why we introduce a more complex method below---this idea will pick up a multiplicative log-star factor}

%\laxman{TODO: switch `paths' with `nodes with a set bit' or `nodes'}
Instead of doing the sum at every level we can do it every $\log\log n$ levels. This will give us a depth of $O(\log n + \log n/\log\log n * \log\log n) = O(\log n)$. In every group of $\log\log n$ levels we will simply take every possible path. The number of paths we have might grow by a factor of $2^{\log\log n} = \log n$. This means that once the number of paths is found to exceed $k$, there may be $\Omega(k \log n)$ paths. Prior to this group of levels, the number of paths was always less than $k$, so the total work is $O(k\log n)$. In this group of levels, we have done at most $O(k \log n)$ work. The number of paths will exceed $k$ at only a single time so this is fine. After this group of levels, if all of these paths were followed to a leaf cluster, the total work could be $\Omega(k \log^2n)$. Instead we will just take $k$ of those paths 
%\quinten{Is this trivial or is this where the compaction thing can be used?}
%\laxman{We can actually afford ordinary (prefix-sum) compaction in this step, since as you point out we only hit this case once and are done afterward.}
and follow them all to a single leaf cluster, so the total work is still $O(k \log n)$.

Finally, at each leaf cluster some edges are returned. If we had at least $k$ paths, each leaf cluster just returns one edge. If not, we can do a prefix sum on the number of level $i$ edges at each leaf cluster, and then return either all of the edges or the first $k$. 
%\quinten{Can we afford storing counters for the number of level $i$ edges at each leaf cluster? I think its fine because the space is proportional to the number of edges $O(m)$.}
%\laxman{Yes, I think it's fine for the reason you mentioned.}
All the edges can be returned in $O(\log n)$ depth and $O(\log n)$ work per edge because they are stored in a balanced binary search tree at each leaf cluster.
This yields the following lemma:

\batchfetch*

\subsection{Restoring the Blocked Invariant}\label{sec:restore}
Here we describe how to restore the blocked invariant in a cluster graph $CG(c)$ that previously maintained the invariant but just had a single level $i$ cluster $P$ split into two \defn{cluster fragments}. A cluster fragment is defined as one of the clusters that formed from the splitting of $P$.

\noindent We define the following operations on a blocked cluster forest:
\begin{itemize}[topsep=0pt,itemsep=0pt,parsep=0pt,leftmargin=15pt]
    \item \textbf{Blocked}$(e)$ determines if a level $i$ edge $e$ is blocked. 
    \item \textbf{SelfLoop}$(e)$ determines if a level $i$ edge $e$ is a self-loop.
    \item \textbf{Merge}$(C_1,C_2)$ merges two level $(i-1)$ clusters $C_1$ and $C_2$. This means the children of $C_2$ are added as children of $C_1$ and $C_2$ is deleted.
    \item \textbf{PushDown}$(e)$ pushes down a level $i$ edge until it is blocked. At every level this finds the clusters containing the endpoints of $e$ and calls $\mathsf{Merge}$ on them. At the end it updates the edge bitmaps appropriately. Returns the merged cluster.
\end{itemize}
%\quinten{Here we may want to explain how Blocked, SelfLoop, and Merge are implemented in $O(\log n)$ time with the local trees and PushDown takes $O(\log n)$ time per level decrease}

In the case where the number of clusters in $CG(c)$ is bounded by some constant, we can afford to restore the invariant by checking that every single cluster is incident to a blocked edge, pushing down any unblocked edge fetched incident to it.
In the rest of this section we focus on the case where the number of clusters in $CG(c)$ is arbitrarily large (e.g. there are $\omega(1)$ clusters).

In this case, Lemmas~\ref{lem:radius} and \ref{lem:heavy_center} tell us that there is a well-defined center cluster that is connected to every other cluster by a blocked edge, and this cluster has the largest size in the cluster graph.
Thus we can naturally think of two cases: $P$ was the center cluster or $P$ was a satellite cluster.
We define a \defn{satellite cluster} as a cluster that is not the center, it was connected to the center with a blocked edge before the split of $P$.

If $P$ was a satellite cluster that split into fragments $P_1$ and $P_2$, then only those two clusters may now violate the blocked invariant. We can restore the invariant by simply fetching and pushing down one \defn{outbound edge} incident to both fragments. We define an outbound edge as an edge incident to a cluster that is not a self-loop. If the outbound edge is blocked, we have certified that that fragment is incident to a blocked edge. If it is unblocked, we push it down thus merging the fragment into another cluster that is incident to a blocked edge.

When $P$ is the center cluster, the blocked invariant may become violated for $P_1$ and $P_2$ themselves, as well as many or all satellite clusters. The structure of this new graph consists of the two centers, $P_1$ and $P_2$, and a set of satellite clusters connected by edges to $P_1$ and/or $P_2$. There may also be edges between $P_1$ and $P_2$ (previously self-loops on $P$), edges between satellites, and self-loops on any cluster.
Figure~\ref{fig:center_split} shows a possible cluster graph before and after the center splitting.

First we prove Lemma~\ref{lem:one_blocked} which will enable us to restore the blocked invariant efficiently.

\begin{lemma}\label{lem:one_blocked}
    Given a cluster graph $CG$ that maintains Invariant~\ref{inv:blocked}, and has $k \geq 4$ clusters, pushing down unblocked edges can cause at most one blocked edge to form between satellites. Additionally, given a cluster graph $CG'$ which was formed from splitting the center of a cluster graph $CG$ that previously maintained Invariant~\ref{inv:blocked} and had $k \geq 4$ clusters, only pushing down unblocked edges between the previous satellites can cause at most one blocked edge to form between the previous satellites.
\end{lemma}
\begin{proof}
    In $CG$, every edge from the center $C$ to any satellite is blocked, and every edge between two satellites is unblocked. The only edges that can become blocked are those between satellites. Assume there are $\geq 2$ blocked edges between satellites. Take any two of these blocked edges. They must be incident to a common cluster or else there is a matching of size two which contradicts Lemma~\ref{lem:blocked_matching_size}. Let these two edges be $(S_1,S_2)$ and $(S_2,S_3)$. We know that there is a blocked edge between $C$ and $S_1$. Then this edge forms a matching of size two with the blocked edge between $S_2$ and $S_3$ which is a contradiction. Therefore there can be at most one blocked edges formed between satellites.

    Now consider $CG'$. The size of every cluster is the same as the size of the corresponding cluster in $CG$, except the two clusters that the center split into. Therefore if only edges between satellites are pushed down, at most one blocked edge can form between them.
\end{proof}

\myparagraph{Overview}
Now we describe how to restore the blocked invariant when a star center is split. The high-level procedure is to fetch an outbound edge incident to each satellite cluster.
We will maintain a running total of the sizes of $P_1$ and $P_2$ combined with the sizes of their adjacent satellites that we have discovered, call these totals $n_1$ and $n_2$.
% \adam{maybe: adjacent? "could" have merged sounds like they are all unblocked}\quinten{Done}
They are initialized as $n(P_1)$ and $n(P_2)$. Intuitively, this simulates what would happen if we did merge all of these satellites into $P_1$ and $P_2$. We also keep an ordered (by size) set of the satellites found incident to $P_1$ and $P_2$, call these $S_1$ and $S_2$.
We will continue to fetch outbound edges for satellites until either $P_1$ or $P_2$ has found two neighboring satellites that could not merge with $P_1$ or $P_2$ if they had merged with all their other neighboring satellites found so far (e.g. there would be two blocked edges incident to $P_1$ or $P_2$ if all other previously examined satellites were merged into their center cluster).
Algorithm~\ref{alg:restore_invariant} shows the pseudo-code for this process.

\begin{algorithm}
\caption{$\mathsf{RestoreInvariant}(CG, i, P_1, P_2)$}
\label{alg:restore_invariant}
\begin{algorithmic}[1]
\State $S_1 \gets \{\}$, $S_2 \gets \{\}, A \gets \{\}$, $n_1 \gets n(P_1)$, $n_2 \gets n(P_2)$
\State $CG.\mathsf{Mark}(P_1)$, $CG.\mathsf{Mark}(P_2)$
\State $CenterTwoBlocked \gets null$ 
\While{true}
%\adam{Thanks. One final doubt is how the marking behaves after edge pushdown starts in Line~\ref{line:satellite_start}. The merged cluster is marked iff $v$ was marked before?}
    \State $X \gets CG.\mathsf{GetUnmarkedCluster}()$
    \State $(u,v) \gets X.\mathsf{FetchOutboundEdge}()$
    % \adam{later you are using FetchOutboundEdge which seems to be internally pushing down self loops. You actually want that here as well, right?} \quinten{Right, I think using that here makes sense.}
    \If{$\mathsf{Cluster}(v, i)$ is $P_1$ or $P_2$} \label{line:center_start} \Comment{Satellite to center}
        \State Insert $X$ into $S_1$ or $S_2$
        \State Increment $n_1$ or $n_2$ by $n(X)$
        \State Insert $(u,v)$ into $A$
        \State $CG.\mathsf{Mark}(X)$ \label{line:center_end}
    \ElsIf{$\neg \mathsf{Blocked}((u,v))$} \Comment{Satellite to satellite} \label{line:satellite_start}
        \State $\mathsf{PushDown}((u,v))$
        \If{$\mathsf{Cluster}(v,i)$ is in $S_1$ or $S_2$}
            \State Increment $n_1$ or $n_2$ by $n(X)$
        \EndIf
    \Else~$CG.\mathsf{Mark}(\mathsf{Cluster}(u))$, $CG.\mathsf{Mark}(\mathsf{Cluster}(v))$
    \EndIf \label{line:satellite_end}
    \If{$n_1 - n(S_1[-1]) > \sizeconst{i} \land n_1 - n(S_1[-2]) > \sizeconst{i}$} \label{line:check_start}
        \State $CenterTwoBlocked \gets P_1$, \textbf{break}
    \EndIf
    \If{$n_2 - n(S_2[-1]) > \sizeconst{i} \land n_2 - n(S_2[-2]) > \sizeconst{i}$}
        \State $CenterTwoBlocked \gets P_2$, \textbf{break}
    \EndIf \label{line:check_end}
\EndWhile
\If{$\neg CenterTwoBlocked$} \label{line:pushA_start}\Comment{two blocked edges not found}
    \For{$(u,v) \in A$}
        \If{$\mathsf{Cluster}(v,i) \notin \{S_1[-1], S_2[-1]\}$}
            \State $\mathsf{PushDown}((u,v))$
        \ElsIf{$\mathsf{Cluster}(v,i) = S_1[-1] \land n_1 \leq \sizeconst{i}$}
            \State $\mathsf{PushDown}((u,v))$
        \ElsIf{$\mathsf{Cluster}(v,i) = S_2[-1] \land n_2 \leq \sizeconst{i}$}
            \State $\mathsf{PushDown}((u,v))$
        \EndIf
    \EndFor \label{line:pushA_end}
    \For{$P \in \{P_1, P_2\}$} \label{line:fix_centers_start}
        \State $(u,v) \gets P.\mathsf{FetchOutboundEdge}()$
        \If{$\neg \mathsf{Blocked}((u,v))$} $\mathsf{PushDown}((u,v))$ \EndIf
    \EndFor \label{line:fix_centers_end}
\Else \Comment{two blocked edges found}
    \State Let $U$ be the center without two blocked edges
    \State Let $S_U$ be the corresponding set for center $U$
    \State $[B,C] \gets \mathsf{FetchAll}(U)$ \label{line:fetch_all}
    \If{$|B| > |A| \land |B| > |C|$} \label{line:pushB_start} \Comment{most edges between centers}
        \For{$(u,v) \in B$} $\mathsf{PushDown}((u,v))$ \EndFor \label{line:pushB_end}
    \Else \Comment{most edges from centers to satellites}
        \For{$(u,v) \in C$} \label{line:pushC} $\mathsf{PushDown}((u,v))$ \EndFor
        \For{$(u,v) \in A$} \label{line:pushA2_start}
            \If{$\mathsf{Cluster}(v,i) \notin \{S_U[-1], S_U[-2]\}$}
                \State $\mathsf{PushDown}((u,v))$
            \EndIf
        \EndFor \label{line:pushA2_end}
        \State $(u,v) \gets \mathsf{Smallest}(CG).\mathsf{FetchEdge}()$ \label{line:fix_smallest_start}
        \While{$\neg \mathsf{Blocked}((u,v))$}
            \State $\mathsf{PushDown}((u,v))$
            \State $(u,v) \gets \mathsf{Smallest}(CG).\mathsf{FetchEdge}()$
        \EndWhile \label{line:fix_smallest_end}
    \EndIf
\EndIf
\end{algorithmic}
\end{algorithm}

If the outbound edge $e$ found incident to a satellite $X$ is incident to $P_1$ (or $P_2$, which is handled analogously), insert $X$ in $S_1$, increment $n_1$ by $n(X)$, and insert $e$ into $A$ (lines~\ref{line:center_start}--\ref{line:center_end}).

If the edge found incident to a satellite $X$ was incident to another satellite $Y$, push it down.
It is possible that this edge is blocked, but that can only happen once according to Lemma~\ref{lem:one_blocked}.
% \adam{is this indicated in the pseudocode currently?}\quinten{Added to pseudo-code.}
If that happens, the invariant has been restored for both of the clusters incident to that edge. We will assume this edge is not blocked.
Then pushing the edge down increases the size of $Y$. If $Y$ was in $S_1$ or $S_2$, we need to increment $n_1$ or $n_2$ respectively by $n(X)$ (lines~\ref{line:satellite_start}--\ref{line:satellite_end}).

After every time a satellite processes an outbound edge, we compare the sizes of the last two elements in $S_1$ and $S_2$ with $n_1$ and $n_2$ to check if the process should stop at this point (lines~\ref{line:check_start}--\ref{line:check_end}).

\myparagraph{Two Blocked Edges Not Found}
If this process never stops, then for each remaining satellite, we have an edge to either $P_1$ or $P_2$. We know that all of these edges can be pushed down with the exception of at most one edge to a satellite in $S_1$ and at most one edge to a satellite in $S_2$ that will be blocked if they are not pushed down. All of these edges will be pushed down (lines~\ref{line:pushA_start}--\ref{line:pushA_end}).
%\adam{hmm... I agree with the above now completely. But the lines~\ref{line:pushA_start}--\ref{line:pushA_end} never attempt to push the "heaviest" edge to $S_1[-1]$... they might be unblocked and possibly should be pushed, right?} \quinten{Good catch. Updated the pseudo-code.}
Now that every satellite is incident to a blocked edge or has merged into a center, we only need to ensure that $P_1$ and $P_2$ are incident to blocked edges. We will just try to fetch a single outbound edge incident to them, pushing it down if possible (lines~\ref{line:fix_centers_start}--\ref{line:fix_centers_end}). This ensures that both centers either found a blocked edge incident to them, merged into the other center, or they are disconnected.
% \adam{Are you sure the above is enough to make all $\leq 4$ remaining clusters adjacent to a blocked edge? It's not obvious to me: $S_1[-1]$ and $S_2[-1]$ may be unblocked. Maybe just say that for all of the remaining $O(1)$ clusters, fetch until blocked?} \quinten{This is enough because if there exists a blocked edge it is incident to either $P_1$ or $P_2$, the other center will find and push down an edge to one of the endpoints of that blocked edge. If there are no blocked edges, there is only $P_1$ and $P_2$ and they will merge if possible. If $S_1[-1]$ and $S_2[-1]$ are unblocked they would have been pushed down. I have written a bit more description to clarify this in the writing.}\adam{thanks!}

\myparagraph{Two Blocked Edges Found}
If we do find two satellites that would form a blocked edge to the same cluster fragment (without loss of generality $P_1$), we stop fetching outbound edges from satellites.
At this point we can be certain that $P_2$ can merge with all of its neighboring satellites without violating the size constraint. Otherwise there would be a matching over blocked edges of size two which contradicts Lemma~\ref{lem:blocked_matching_size}.

Now we will fetch every edge incident to $P_2$ (line~\ref{line:fetch_all}).
Ignoring self-loops, the remaining edges of those we have fetched so far fall in to three sets: ($A$) outbound edges from satellites previously fetched, ($B$) edges between $P_1$ and $P_2$, and ($C$) edges from $P_2$ to satellites.
We will ensure that the largest of these 3 sets is pushed down. This means at least a constant fraction of the edges were pushed down so the work done is $O((k+1)\log n)$ where $k$ is the total number of edges pushed down.

If $B$ was the largest set, we push it down (lines~\ref{line:pushB_start}--\ref{line:pushB_end}). At this point we have restored the size of the original star center, so every cluster must be incident to a blocked edge and we are done.
If $A$ or $C$ was the largest set, we start by pushing down all of $C$ (line~\ref{line:pushC}). Then we push down all of the edges in $A$ that were not also in $C$, besides the at most two blocked edges (lines~\ref{line:pushA2_start}--\ref{line:pushA2_end}). Therefore even if $A$ was the largest, we still have pushed down at least a constant fraction of the edges.

Now that $P_2$ has merged with all of its satellites, either it is disconnected from the rest of the cluster graph, or it has become a satellite of $P_1$, restoring a star structure.
There may still be some satellites that are not incident to any blocked edge because we did not fetch outbound edges to them earlier. Now we will continue fetching outbound edges, but this time we will fetch them incident to the smallest cluster, and push them down immediately until a blocked edge is encountered (lines~\ref{line:fix_smallest_start}--\ref{line:fix_smallest_end}).
Lemma~\ref{lem:smallest_blocked} proves that if the smallest cluster is incident to a blocked edge, then the rest of the clusters must also be, thus completing the restoration of the blocked invariant.
In Appendix~\ref{sec:modified_local_trees} we describe how the local trees can be modified to return the smallest cluster in a cluster graph in $O(\log n)$ time.

\begin{lemma}\label{lem:smallest_blocked}
    In a cluster graph with $k \geq 4$ level $i$ clusters where every cluster has an edge to a single center cluster $C$ and at least two edges incident to $C$ are blocked, if the smallest non-center cluster $S$ is incident to a blocked edge, every cluster is incident to a blocked edge.
\end{lemma}
\begin{proof}
    Since $S$ is incident to a blocked edge, there exists a cluster $X$ neighboring $S$ such that $n(S) + n(X) > \sizeconst{i}$. If $X$ was the center $C$, then the edge between each satellite $Y$ and $C$ must be blocked since $n(Y) \geq n(S)$, so each cluster is incident to a blocked edge. We will show that $X$ could not be another satellite by contradiction. There are at least two blocked edges between $C$ and the largest two satellites. The blocked edge between $S$ and $X$ would form a matching of size two over the blocked edges with the edge from the center and at least one of the two largest satellites, which is a contradiction.
\end{proof}

\section{Additional Proofs in Section~\ref{sec:parallelupdate} and ~\ref{sec:batchpar}}

\subsection{Pushing Down Edges}

Let $X$ be the level $(i-1)$ cluster that is not isolated (or any arbitrary cluster if they are all isolated). We can combine all of the local trees by using the $\mathsf{BatchInsert}$ operation of the modified local tree to insert all of the single level $(i-2)$ clusters into the local tree for $CG(X)$. Then we use $\mathsf{BatchDelete}$ to delete all of the level $(i-1)$ clusters other than $X$ in the local tree for the level $i$ parent cluster, and increment $n(X)$ by the combined size of all of the isolated clusters. This yields the following lemma which indicates an efficient implementation of \textbf{PushDownGroup}:

\myparagraph{Batch Push Down}
Now we describe how to implement the \textbf{BatchPushDown} operation using calls to \textbf{PushDownGroup}. The general strategy will be to take a spanning forest of the edges and decompose it into disjoint stars. Then for each star we will merge into the center a maximal prefix of the clusters sorted by size that can be merged into the center without violating the size constraint.

The first step is to compute a spanning forest over the $k$ edges in $E$. To convert the list of edges into a graph representation, we need to determine the cluster that the endpoints of each edge of $E$ are contained in. This step takes $O(k \log n)$ work and $O(\log n)$ depth. The spanning forest computation can then be done in $O(\log k)$ depth and $O(k\log k)$ work~\cite{shiloach1982cc}.
A spanning forest is sufficient because the operation only requires that the clusters containing an endpoint of any edge in $E$ follow the blocked invariant. The spanning forest has at least one edge incident to every such cluster.

Next we define some arbitrary cluster as the root of each tree in the spanning forest, and use the Euler tour technique to compute the level of each cluster relative to its root. This takes $O(\log k)$ depth and $O(k)$ work~\cite{tarjan1984ett}.
Using the levels of each cluster we can partition the rooted trees into disjoint stars where the even level clusters are centers and the odd level clusters form stars around their parent in the rooted tree.

%\adam{Shouldn't you make the ``center levels'' either the even or the odd depending on which has smaller total size?} \quinten{I don't think this is needed because the depth relies on the number of levels being halved each round. Each edge is still inspected only once so the work is also fine.}\adam{cannot it still happen that even in the first iteration, the spanning tree is such that every odd-level cluster's edge to the parent is already blocked, but there exist some even-level leaf clusters? Then I guess the iteration will not make any progress, and in the next iteration nothing will change: the same vertices will be odd- and even levels.}

For each star we will take a maximal prefix of the satellites in sorted order by size that can be merged into the center without violating the size constraint. This will ensure that either all of the satellites have merged into the center, or every remaining satellite has a blocked edge to the center.
Sorting the clusters by size and taking a prefix sum around all stars will take $O(\log k)$ depth and $O(k \log k)$ work~\cite{cole1986sort}\cite{cole1989prefix}.
Finally, for each star we will push down the set of edges determined in the previous step. We can use the \textbf{PushDownGroup} operation to do this in $O(\log k)$ depth and $O(k \log n)$ total work across all stars.

Doing this once ensures that the blocked invariant holds for every odd level cluster, or it has been merged into its center. Now we need to repeat the process multiple times until this holds for all clusters. First we update the spanning forest to reflect the edges that were pushed down previously (this can be done using the Euler tour technique also). Then we again do a star decomposition and merge a maximal prefix of the smallest satellites into each center.
%\adam{so where is the outer loop placed? every iteration, a new spanning forest is computed after contraction?} \quinten{I think it makes sense to use the same spanning forest but contracting the edges that were pushed down. Updated the description here.}
Each round halves the number of levels that may still violate the blocked invariant. We repeat this process $O(\log n)$ times to fully restore the blocked invariant for each cluster.

Overall the depth is $O(\log^2n)$. Each edge is processed only in a single round, so in total there is $O(k \log n)$ work.
This yields Lemma~\ref{lem:batch_push_down} which provides an efficient implementation for the \textbf{BatchPushDown} operation.

\begin{lemma} \label{lem:batch_push_down}
    Given a set $E$ of $k$ level $i$ edges, enforcing that every level $(i-1)$ cluster containing an endpoint of any edge in $E$ is incident to a blocked edge or its parent is isolated, can be done in $O(\log^2n)$ depth and $O(k \log n)$ work.
\end{lemma}

\subsection{Restoring the Blocked Invariant in Parallel}
\myparagraph{Restoring the Blocked Invariant}
When a cluster is split during deletion, the algorithm must restore the blocked invariant in that cluster graph before continuing.
If the split cluster was a satellite, we just need to find and/or push down a single outbound edge incident to both fragment clusters. Lemma~\ref{lem:doubling_fetch} proves that we can find such an edge in $O(\log^2 n)$ depth and $O(k \log n)$ work which can be charged to the self-loops that are found and pushed down.

If the split cluster was the center, the sequential strategy was to fetch an outbound edge incident to each satellite until we found that two blocked edges incident to one of the fragment clusters, or we found an outbound edge from each satellite.
In parallel, we will start considering satellites in $O(\log n)$ rounds of geometrically increasing size. Each satellite can fetch an outbound edge in $O(\log n)$ rounds of doubling search over its incident edges. Similar to the technique used by Acar et al.~\cite{acar2019parallel} we can interleave the rounds of doubling edge search with the rounds of increasing number of satellites.
%\adam{what does this mean exactly?} \quinten{Basically you start with 1 cluster that fetches 1 edge, then 2 clusters that both fetch 2 edges, then 4 clusters that all fetch 4 edges, etc. Each round the work multiplies by 4 so it can be charged to the previous round still. I'm not sure how much detail to go into this since it is an idea from the previous paper and it will make things significantly more complicated.}
%\adam{Then I guess some kind of intuition like this is still helpful, and perhaps a more specific reference (e.g. which section) to~\cite{acar2019parallel}?}
This allows us to do this in $O(\log n)$ rounds for a total depth of $O(\log^2 n)$.
% \quinten{parallel set operations}
We can track the edges found incident to each center fragment using a concurrent set data structure.

If this process ends with neither center fragment finding two blocked edges, we just push down all of the unblocked edges to center fragments. We can use two calls to \textbf{PushDownGroup} to do this in $O(\log n)$ depth.
Then we need to find an outbound edge incident to both $P_1$ and $P_2$, which can be done in $O(\log^2 n)$ depth by Lemma~\ref{lem:doubling_fetch}.

If two blocked edges were found incident to a center fragment, then we fetch all of the edges incident to the other center fragment. Lemma~\ref{lem:doubling_fetch} proves that this can be done in $O(\log^2 n)$ depth.
Then we push down some edges, all of which are incident to $P_1$ or $P_2$, so we can use two calls to \textbf{PushDownGroup} to do this in $O(\log n)$ depth.
Finally, we need to ensure that the smallest cluster is incident to a blocked edge. We can do this by fetching geometrically increasing sized prefixes of the smallest clusters, and finding an outbound edge for all of them. Once again we can interleave the rounds of doubling edge search with the rounds of increasing number of clusters to achieve $O(\log^2 n)$ depth.

\subsection{Analysis of Parallel Deletion}

\myparagraph{Cost Analysis}
During both the upward and downward sweep, the depth is $O(\log^2n)$ per level, yielding an overall depth of $O(\log^3n)$.

To analyze the work in restoring the blocked invariant we prove the following two lemmas that bound the work done during the search for outbound edges from satellites, and the work done to ensure that the smallest cluster is incident to a blocked edge:

\begin{lemma} \label{lem:par_sat_work}
    During the interleaved doubling edge search over the satellites, the work done is $O((x+y)\log n)$ where $x$ is the number of edges fetched that can be pushed down, and $y$ is the number of edges fetched in all but the last round that cannot be pushed down.
\end{lemma}
\begin{proof}
    During the interleaved doubling edge search over satellites, the number of edges we fetch increases geometrically each round. Let $t$ be the total number of edges fetched across all rounds, then the work done is $O(t \log n)$.
    Once the search ends, we know that all of the edges fetched in previous rounds (excluding the final round) can be pushed down except for at most $y$ edges.
    The number of edges fetched in previous rounds is at least a constant fraction of the edges fetched in the last round. So the number of edges that can be pushed down is $x \geq c \cdot t - y$, for some constant $0 < c < 1$. Then the work done is $O((x+y)\log n)$.
\end{proof}

\begin{lemma} \label{lem:par_smallest_work}
    Ensuring that the smallest cluster in a cluster graph $CG(c)$ is incident to a blocked edge can be done in $O(\log^2 n)$ depth and $O((x+1)\log n)$ work where $x$ is the number of edges pushed down by this process.
\end{lemma}
\begin{proof}
    The modified local trees support getting a prefix of the $k$ smallest clusters in $O(\log n)$ depth and $O(k \log n)$ work. The total depth is $O(\log^2 n)$ because we interleave the rounds of doubling prefix size with the rounds of fetching edge batches of doubling size. Each round takes $O(\log n)$ depth.
    The number of edges we fetch increases geometrically each round. Let $t$ be the total number of edges fetched across all rounds, then the total work done is $O(t \log n)$. Each cluster fetches at least one edge so the work in getting prefixes is accounted for.
    The search ends when the first blocked edge is found. The number of edges fetched in previous rounds is at least a constant fraction of the edges fetched in the last round. All of the edges in the previous rounds can be pushed down, unless the blocked edge was found in the first round. So the number of edges that will be pushed down is $x \geq c \cdot t - 1$, for some constant $0 < c < 1$. Then the work done is $O((x+1)\log n)$.
\end{proof}

During the parallel search for outbound edges from satellites, at most $y=3$ edges can not be pushed down (other than the last round), at most one to both center fragments, and at most one between satellites. So applying Lemma~\ref{lem:par_sat_work}, the work during the search for outbound edges from satellites is $O((x_0+1)\log n)$ where $x_0$ is the number of edges that can be pushed down.

Let's consider the work when two blocked edges were found. Let $t_1$ be the number of edges fetched when fetching all of the edges incident to $P_2$. Then the total work is $O(t_1 \log n)$ by Lemma~\ref{lem:doubling_fetch}.
Then when we push down the larger of sets $A$, $B$, and $C$, and possibly some edges in the other sets. Let $x_1$ be the number of edges pushed down from these sets. Then $x_1 \geq 1/3 \cdot t_1$, The work done was $O(x_1 \log n)$, and the work to push them down is the same.
At this point $O(x_0 + x_1)$ edges have been pushed down, since $x_1 \geq x_0$.
Finally, applying Lemma~\ref{lem:par_smallest_work}, the work done in the process to ensure the smallest cluster is incident to a blocked edge is $O(x_2 \log n)$ where $x_2$ is the number of edges pushed down.
So the total work in the case where two blocked edges were found is $O((x_0 + x_1 + x_2 + 1) \log n) = O((x+1)\log n)$ where $x = \Omega(x_0 + x_1 + x_2)$ is the total number of edges pushed down.

Now we consider the work when two blocked edges were not found. Pushing down the $x_1$ edges fetched from satellites takes $O(x_1 \log n)$ work.
Let $t_3$ be the total number of edges fetched when fetching an outbound edge incident to $P_1$ and $P_2$. The work done is $O(t_3 \log n)$ by Lemma~\ref{lem:doubling_fetch}. Each edge fetched is either an outbound edge or a self-loop. All of the self-loops will be pushed down, and there are $2$ outbound edges. So the number of edges pushed down is $x_3 \geq t_3 - O(1)$, so the work is $O((x_3 + 1)\log n)$.
Finally, possibly pushing down the two outbound edges may take $O(\log n)$ work.
So the total work in the case where two blocked edges were not found is $O((x_0 + x_3 + 1) \log n) = O((x+1)\log n)$ where $x = \Omega(x_0 + x_3)$ is the total number of edges pushed down.

Of the $O((x+1)\log n)$ work done to restore the blocked invariant per level, we can charge $O(x \log n)$ of it to the $x$ edges pushed down, leaving $O(\log n)$ uncharged work. The rest of the work in the upward sweep is also $O(\log n)$ per level. Across all levels, the $O(\log^2 n)$ work can be charged to the deletion itself.

During the downward sweep we perform $O(|E_\ell| \log n)$ work per level which can be charged to the edges being pushed down from the previous level.
This yields the following theorem for edge deletion in the blocked cluster forest:

\begin{theorem} \label{thm:par_del}
    A deletion in the blocked cluster forest can be done in $O(\log^3n)$ depth and $O(k\log^2n)$ work where $k$ is the total number of times an edge is pushed down during deletion.
\end{theorem}

\subsection{Restoring the Blocked Invariant for Batches}\label{sec:batchrestore}
Here we describe how to restore the blocked invariant in a cluster graph $CG(c)$ that previously maintained the invariant, but just had some number of clusters split into $k$ cluster fragments, $P=\{P_0,...,P_{k-1}\}$. This will result in the proof of Lemma~\ref{lem:batch_restore_invariant}.

\batchrestore*

Similar to the sequential algorithm for restoring the blocked invariant (Section~\ref{sec:restore}), we will treat each cluster fragment differently depending on if it was previously part of a satellite cluster or a center cluster.
If there was no cluster fragment that came from the center cluster of $CG(c)$ splitting, then the conditions for Lemma~\ref{lem:smallest_blocked} hold. Then we can use the technique described in Section~\ref{sec:pardel} to ensure that the smallest cluster is incident to a blocked edge, thus ensuring that all clusters are.

If there was a cluster fragment that came from the center cluster of $CG(c)$ splitting, then the process for restoring the blocked invariant will ensure every cluster is incident to a blocked edge, so we do not have to consider cluster fragments from satellite clusters separately.
Therefore we will assume that all of $P=\{P_0,...,P_{k-1}\}$ are fragments of the original center cluster.

Now we will prove an extension to Lemma~\ref{lem:one_blocked} that applies to the case where the center cluster was split into multiple fragments.

\begin{lemma}\label{lem:batch_one_blocked}
    Given a cluster graph $CG'$ which was formed from splitting the center of a cluster graph $CG$ (possibly into multiple parts) that previously maintained Invariant~\ref{inv:blocked} and had $k \geq 4$ clusters, only pushing down unblocked edges between the previous satellites can cause at most one blocked edge to form between the previous satellites.
\end{lemma}
\begin{proof}
    The proof is identical to that of Lemma~\ref{lem:one_blocked}.
\end{proof}

The process for restoring the blocked invariant will begin exactly the same as the parallel process for restoring the invariant with a single split as described in Section~\ref{sec:pardel}.
We fetch outbound edges incident to satellites, doubling the number of satellites we look at in each round and interleaving the rounds of edge searching with the rounds of increasing numbers of satellites.
In the batch case, for each cluster fragment $P_i$, we store a set $S_i$ and total size $n_i$.
the search may stop if any of the $k$ fragments finds that it could form two blocked edges incident to it.

If edges are found between two satellites, we can once again push them down immediately. Now Lemma~\ref{lem:batch_one_blocked} gives us the same guarantee that only one blocked edge may form between the satellites.
For edges found to a center fragment $P_i$, we add the satellite to $S_i$ and update $n_i$ just as described previously.

\myparagraph{Two Blocked Edges Found}
If we found that two blocked edges could be formed incident to, without loss of generality, $P_0$, then we know that every other center fragment can merge with its entire neighborhood of satellites.
For each of those center fragments, we fetch all of the edges incident to it. Ignoring self-loops, the remaining edges of those we have fetched so far fall in to three sets: ($A$) edges between center fragments, ($B$) edges from center fragments other than $P_0$ to satellites, and ($C$) outbound edges from satellites previously fetched.
Once again, we will ensure that the largest of these 3 sets is pushed down, thus at least a constant fraction of the edges are pushed down.

If $A$ was the largest set, it is pushed down. Since we fetched all of the edges incident to each center fragment other than $P_0$, we have found every edge that exists between two center fragments. At this point some, all, or none of the center fragments may have merged into $P_0$.
For those that did not, we will push down all of its incident edges (previously fetched in $B$).
Then for the component with center $P_0$ we will use the same process as Section~\ref{sec:pardel} to ensure that the smallest cluster is incident to a blocked edge, thus restoring the invariant for all clusters by Lemma~\ref{lem:smallest_blocked}.

If $B$ or $C$ was the largest set, we start by pushing down all of $C$, except the two blocked edges between $P_0$ and satellites in $S_0$. Then we push down all of the edges in $B$ that were not also in $C$.
Now each center fragment has contracted with all of its satellites, either forming a single disconnected cluster, or becoming a satellite of $P_0$. We can determine this by attempting to get a single outbound edge incident to them. This may produce $O(k \log n)$ uncharged work if they each find a blocked edge.
Then for the component with center $P_0$ we will use the same process as Section~\ref{sec:pardel} to ensure that the smallest cluster is incident to a blocked edge again, thus restoring the invariant for all clusters by Lemma~\ref{lem:smallest_blocked}.

\myparagraph{Two Blocked Edges Not Found}
In this case we have found an outbound edge incident to each satellite. At most $k$ of them can be blocked, one for each center fragment. Thus there may be $O(k \log n)$ uncharged work for fetching those $k$ edges that will not be pushed down.
We will push down all of these edges that can be. The remaining clusters are the $k$ center fragments and at most one satellite $Q$. There cannot be two (or more) satellites $Q_1$ and $Q_2$ that had blocked edges incident to them because each center fragment only had 1 blocked edge found incident to it, so that would imply a matching of size two over the blocked edges, violating Lemma~\ref{lem:blocked_matching_size}.

To finish restoring the invariant we can now find an outbound edge incident to each of the $O(k)$ clusters remaining, and then call \textbf{BatchPushDown} on those edges. This will take $O(\log^2 n)$ depth, $O(k \log n)$ work, and ensure that all of the clusters are incident to a blocked edge.

\myparagraph{Cost Analysis}
The depth of the interleaved search for outbound edges from satellites is $O(\log^2 n)$.
In the case where two blocked edges are found, finding all edges incident to all but one of the center fragments can be done in parallel with depth $O(\log^2 n)$ by Lemma~\ref{lem:doubling_fetch}, pushing down all of $A$, $B$, and or $C$ can be done using \textbf{BatchPushDown} in $O(\log^2 n)$ depth, getting an outbound edge out of each center fragment takes $O(\log^2 n)$ depth by Lemma~\ref{lem:doubling_fetch}, and the process to ensure the smallest cluster is incident to a blocked edge takes $O(\log^2 n)$ depth.
In the case where two blocked edges are not found, pushing down the fetched edges can be done using \textbf{BatchPushDown} in $O(\log^2 n)$ depth, fetching an outbound edge incident to each remaining cluster can be done in $O(\log^2 n)$ depth by Lemma~\ref{lem:doubling_fetch}, and calling \textbf{BatchPushDown} on these edges takes $O(\log^2 n)$ depth.
Thus the depth to restore the blocked invariant is $O(\log^2 n)$ as stated in Lemma~\ref{lem:batch_restore_invariant}.

Lemma~\ref{lem:par_sat_work} tells us that the work done during the search for outbound edges from satellite is $O((x_0+k)\log n)$ where $x_0$ is the number of edges found that can be pushed down.

Let's consider the work when two blocked edges were found. Let $t_1$ be the total number of edges fetched when fetching all of the edges incident to all but one of the center fragments. Then the total work is $O(t_1 \log n)$ by Lemma~\ref{lem:doubling_fetch}.
Then when we push down the larger of sets $A$, $B$, and $C$, and possibly some of the other sets. Let $x_1$ be the number of edges pushed down from these sets. Then $x_1 \geq 1/3 \cdot t_1$, The work done was $O(x_1 \log n)$, and the work to push them down is the same.
At this point $O(x_0 + x_1)$ edges have been pushed down, since $x_1 \geq x_0$.

Let $t_2$ be the total number of edges fetched when fetching an outbound edge out of each center fragment. Then the total work is $O(t_2 \log n)$ by Lemma~\ref{lem:doubling_fetch}.
Each edge fetched is either an outbound edge or a self-loop. All of the self-loops will be pushed down, and there are at most $k$ outbound edges. So the number of edges pushed down is $x_2 \geq t_2 - k$, so the work is $O((x_2 + k)\log n)$.

Finally, Lemma~\ref{lem:par_smallest_work} tells us that the work done in the process to ensure the smallest cluster is incident to a blocked edge is $O(x_3 \log n)$ where $x_3$ is the number of edges pushed down.
So the total work in the case where two blocked edges were found is $O((x_0 + x_1 + x_2 + x_3 + k) \log n) = O((x+k)\log n)$ where $x = \Omega(x_0 + x_1 + x_2 + x_3)$ is the total number of edges pushed down.

Now we consider the work when two blocked edges were not found. Pushing down the $x_0$ edges fetched from satellites takes $O(x_0 \log n)$ work.
Let $t_4$ be the total number of edges fetched when fetching an outbound edge incident to each of the $O(k)$ remaining clusters. The work done is $O(t_4 \log n)$ by Lemma~\ref{lem:doubling_fetch}. Each edge fetched is either an outbound edge or a self-loop. All of the self-loops will be pushed down, and there are $O(k)$ outbound edges. So the number of edges pushed down is $x_4 \geq t_4 - O(k)$, so the work is $O((x_4 + k)\log n)$.
Finally, calling \textbf{BatchPushDown} on the $O(k)$ outbound edges takes $O(k \log n)$ work.
So the total work in the case where two blocked edges were not found is $O((x_0 + x_4 + k) \log n) = O((x+k)\log n)$ where $x = \Omega(x_0 + x_4)$ is the total number of edges pushed down.
This yields the proof of Lemma~\ref{lem:batch_restore_invariant}.

\subsection{Analysis of Batch Updates}

During both the upward and downward sweeps for deletion, the depth is $O(\log^2 n)$ per level, yielding an overall depth of $O(\log^3 n)$.

To analyze the work first note that the size of the active set of components will remain $O(k)$ because there can be at most $k+1$ components caused by $k$ edge deletions.
During the upward sweep, Lemma~\ref{lem:batch_restore_invariant} tells us that restoring the blocked invariant in a cluster graph with $k'$ active clusters takes $O(k' + (x+1)\log n)$ work where $x$ is the total number of edges pushed down by this process. The uncharged work is $O(k' + \log n)$. Across all the cluster graphs at a given level where $\sum k' = O(k)$, this work is $O(k \log n)$. Across every level this results in $O(k \log^2 n)$ uncharged work. We can charge $O(\log^2 n)$ of this work to each of the $k$ edge deletions.
The parallel connectivity search and splitting of clusters takes $O(\log^2 n)$ depth per level and $O(k' \log n)$ uncharged work per cluster graph. Just like the previous argument, across all cluster graphs at every level the total uncharged work is $O(k \log^2 n)$. We can charge $O(\log^2 n)$ of this work to each of the $k$ edge deletions.
During the downward sweep we perform $O(|D_\ell| \log n)$ work per level which can be charged to the edges being pushed down from the previous level, similar to the analysis of batch insertion.

The main result of this section is the following theorem for batches of updates in the blocked cluster forest:

\batchupdate*
\section{Additional Implementation Details and Experiments}

\subsection{CF Implementation and Optimizations}

We assume the reader is familiar with the CF algorithm as described in Wulff-Nilsen's paper~\cite{wulff2013faster} and also in Section~\ref{sec:seqcf}.
We decompose our implementation into several pieces, starting bottom-up, from the design of {\em leaf nodes}, then, {\em rank and local trees}, and finally describe our {non-tree edge optimization} that enables us to safely delete many edges without performing replacement edge search. We also describe a {\em two-level queue} optimization that decreases the number of traversals of the cluster forest.

\myparagraph{Leaf Design}
A $leaf$ stores all the edges incident to a vertex. It has the following attributes:
\begin{itemize}
    \item $id$, a unique integer identifier of a vertex in the graph
    \item $bitmap$, which is a 64-bit bitset where the $i$-th bit is set if and only if this leaf (i.e., this vertex) has level-$i$ incident edges.
    \item A collection of {\em edge sets}, each corresponding to a different level.
    Each edge set stores all edges belonging to this level that are incident to this leaf. 
\end{itemize}

Theoretically, these edge sets should be stored in a nested map where the key of the outer map is the level, and the value is a set which contains all edges at this level.  
Notice that this approach requires $O(\log n)$ time to locate the specific set that contains level $i$ edges. 
Since the algorithm mostly operates on the {\em inner edge set}, lookups on the outer map can become a bottleneck. Therefore, we replace the outer map with a vector. 
Importantly, we can implement this change without sacrificing the time or space complexity of the data structure. 

Our approach replaces the outer map with a vector as follows: the vector stores a collection of pairs in which the first element is the level, and the second element is a pointer to the edge set that stores the edges corresponding to that level. 
To quickly locate the index of the level $i$ edge set, we can simply take the bitmap for this leaf, zero-out all bits corresponding to levels larger than $i$ (this can be done in $O(1)$ time by performing a bit-shift), and finally counting the number of set bits in this modified bitmask (e.g., using \textsc{popcount}).
This entire operation can be done in $O(1)$ time using bit-operations on words, and is extremely fast, and yields the edge set corresponding to the level $i$ edges.
We note that operations on the collection of edge sets that introduce a new edge set or deleting an empty edge set still require $O(\log n)$ time, matching the theoretical complexity of the naive solution.
However, many operations that work solely within already-existing edge sets can be sped up significantly in practice, although this does not yield an asymptotic improvement in the running time.

%The total number of set bits in this $bitmap$ minus 1 will be the index of the pair that contains the information of level-$i$ edges in the vector. 
%
%Algorithm 4 shows how insertion and deletion of a level-$i$ edge $(u,v)$ work in the leaves.

%\begin{algorithm}
%\caption{$\mathsf{Edge\ Insertion\ and\ Deletion \ In\ Leaf}$}
%\label{alg:batchins}
%\begin{algorithmic}[4]
%    \Procedure{getIndex}{$i,bitmap$} \Comment{get index of pair}
%        \State $b\gets$ $bitmap$ << ($bitmap$.$size()- 1$ - $i$)
%        \State \textbf{return} count\_set\_bits($b$) - $1$
%    \EndProcedure
%    \Procedure{insert}{$u,v,i$}\Comment{insert v to the  $edge\_sets$ of u}
%        \State $index$ $\gets$ getIndex($i$,Leaf[u]->$bitmap[i]$)
%         \If{$\mathsf{Leaf[u]}$->$bitmap[i]$ is $true$}
%            \State Leaf[u]->$edge\_sets$[$index$]->$insert(v)$
%        \Else
%            \State Leaf[u]->$bitmap[i]\gets$ $true$
%            \State $nghs$ $\gets$ new $set\{\}$
%            \State $nghs$.$push\_back(v)$
%            \State $pair$ $\gets$ make\_pair($index$,$nghs$)
%            \If{Leaf[u]->$edge\_sets.size < index+1$}
%                \State Leaf[u]->$edge\_sets$.$push\_back$($pair$)
%            \Else
%                
%            \EndIf
%        \EndIf
%        
%    \EndProcedure
%\end{algorithmic}
%\end{algorithm}

%\myparagraph{Rank- and Local-Tree Design}

\myparagraph{Non-Tree Edge Tracking Optimization}
%(we get $O(\log n)$ time per non-tree deletion; should mention that).
%
If one considers the HDT algorithm~\cite{holm2001poly} and the earlier Henzinger-King algorithm~\cite{henzinger1995randomized}, both algorithms work by maintaining a spanning forest consisting of \defn{tree edges} that serves as a certificate for the connectivity of the graph.
Both algorithms store this spanning forest in an Euler Tour Tree which is used to answer queries for whether two vertices are connected.
Perhaps more importantly, when processing deletions, both algorithms observe that deletions of edges that are {\em not in this spanning tree}, i.e., of \defn{non-tree edges} are essentially free, since deleting such an edge does not affect the connectivity of the overall graph.

Unfortunately in the cluster forest data structure, there is no clear notion of a tree or non-tree edge, and obtaining a similar speedup as in the HDT/HK algorithms by avoiding work when processing non-tree edge deletions is not at all obvious.
One kind of edge in the cluster forest that can be deleted cheaply is a \defn{self-loop edge}, a level $i$ edge such that the level $(i-1)$ cluster representatives of both $u$ and $v$ are identical.
However, testing whether an edge is a self-loop edge requires $O(\log n)$ time in order to traverse the cluster forest and find the level $(i-1)$ cluster representatives of the level $i$ edge.
However, in the HDT algorithm checking if an edge is a non-tree edge takes $O(1)$ time since this information is stored in a hash-table. This fact makes HDT much more practical on dense graphs where most of the deletions target non-tree edges.

Before introducing our optimization, it is instructive to think about a simplistic approach to identify tree vs. non-tree edges. Suppose that we mark the first insertion that connects two components as a tree edge, and all subsequent insertions that don't affect connectivity as non-tree edges.
Now, suppose we delete some other tree edge and perform a replacement edge search.
During the replacement edge search, the algorithm may end up pushing down some of the non-tree edges that are parallel to a tree edge (thus merging the clusters that this tree edge goes between), while leaving this tree edge at its original level.
Subsequently, one of these pushed-down non-tree edges could get deleted, but since this edge is marked as a non-tree edge, we would not perform a replacement search, and could leave the cluster-graph at a lower level disconnected (also note that the tree edge becomes a self-loop at its unmodified level, and the notion of doing a ``replacement search'' is ill defined since both level $(i-1)$ representatives of the deleted edge are actually the same).
This example illustrates that any tree/non-tree marking strategy should be extremely careful to ensure that each cluster graph at every level $i$ is {\em connected}.
%
% We note that we actually ran into correctness issues like the ones described above when trying to come up with a correct optimization.

Our optimization based on this idea is called \defn{non-tree edge tracking}, and works as follows.
\begin{itemize}
\item  When inserting an edge $(u,v)$ that connects two different components, we mark this edge as a tree edge.
\item  Otherwise, if the edge goes within the same component, we mark this edge as a non-tree edge.
\end{itemize}
This part of our optimization is extremely natural. What is significantly more interesting is how to update tree/non-tree edge designations during a deletion and replacement edge search.
\begin{itemize}
\item If an edge $e=(u,v)$ is deleted and $e$ is marked as non-tree, we simply delete it, and update the bitmaps for $u$ and $v$.
\item Otherwise, if $e=(u,v)$ is a level $i$ tree edge, we perform replacement edge search from the level $(i-1)$ representatives $C_u$ and $C_v$. When performing the replacement edge search, for each level $(i-1)$ cluster found, we mark the first edge that reaches that cluster as a tree edge. We note that this marking occurs in both of the searches that we perform---both the search out of $C_u$ and the search out of $C_v$.
\end{itemize}

A reader may wonder whether the amount of edges marked as tree is much larger than necessary.
Although we are not sure whether this approach is tight in a theoretical sense (i.e., whether we may be over-marking edges as tree), the important property of our optimization is that it is {\em correct}, and that it is also highly effective in practice and empirically does not seem to overmark too many edges unnecessarily as tree edges (as we mentioned in the main body of the paper, there is a few percentage difference in the number of non-tree edges we detect compared with HDT).

To show correctness, it suffices to prove that the level $i$ cluster graphs (i.e., the graphs consisting of level $(i-1)$ cluster nodes connected by their level $i$ edges) are connected when we only include the tree edges.
Assume that this property holds before a replacement search, and consider the changes that we make after performing replacement search at some level $i$ (using level $i$ edges).
Suppose without loss of generality that we push down the set of vertices searched from $C_u$.
Since the first edge to each of these vertices is marked as a tree edge and is also pushed down, this entire set of vertices will be reachable via level $(i-1)$ edges in the level $(i-1)$ cluster graph.
Why do we mark the first-visit edges incident from the other search (i.e., out of $C_v$, the search that does not get pushed down) as tree?
To see why this is also necessary, consider what happens if the searches connect with each other and $C_v$ connects to $C_u$'s search via a non-tree edge---the issue in a nutshell is that the level $i$ cluster graph now containing $C_v$ may not be connected only using tree edges.
Our strategy of also marking all first-visit edges encountered in $C_v$'s search as tree handles this case, since any non-tree edge that is used to reconnect the two components will correctly be marked as tree in the level $i$ cluster graph.
We note that a more sophisticated version of this optimization could use some more complicated data structures to try and identify {\em which} of the first-visit edges are actually necessary to add, but we found this simpler approach of just marking all first-visit edges to work remarkably well in practice and found these more complicated ideas to have non-trivial overheads that slowed down replacement edge search. 
Exploring this rich optimization space further is an interesting direction for future work.

%
%Both D-Tree and HDT maintain the spanning forest of the graph. We call the edges in the spanning forest $tree$ edge. And the rest of edges in the graph are so called $non-tree$ edge 
%%

\myparagraph{Two-Level Queue Optimization}
Our last optimization also helps eliminate some unnecessary traversals of the cluster forest hierarchy and thus improves the running time of our implementation.
%
%Although we do not ablate this optimization in the paper due to space constraints, we feel it is important to describe here for completeness since it was an early and important part of our design.
%
%
The motivation for the optimization is the fact that fetching edges in the cluster forest implementation is a major bottleneck.
Note that in an HDT implementation, a fetched level $i$ edge can be pushed down immediately if it is not a replacement edge (notice that no clusters need to be merged in HDT).
By contrast, in a cluster forest implementation pushing a level $i$ edge down to level $i-1$ is not at all a simple task since we have two alternating searches going on, and we don't know which of the two sets should be pushed down until the end of the searches.
However, modifying the cluster forest seems to be necessary (e.g., updating the bitmaps for the endpoints of this edge) since otherwise, the fetch operation cannot find the next level $i$ edge to check.

Our \defn{two-level queue} optimization speeds up the process of fetching level $i$ edges in the cluster forest by maintaining two separate queues of objects. The first {\em cluster queue} consists of level $(i-1)$ clusters that we visit when traversing the cluster graph. The second queue which we call the {\em leaf queue} consists of leaves that we visit which have incident level $i$ edges.
The replacement search procedure maintains two separate two-level queues for $C_u$ and $C_v$'s searches, respectively.
During a search, we first check the leaf queue. If the leaf-queue is non-empty, we take the first leaf, extract a level $i$ edge from this leaf, and traverse it to find the neighboring leaf and corresponding level $(i-1)$ cluster.
If the neighboring leaf and cluster are not already in their respective queues, we insert them into the appropriate queue.
To avoid extracting the same edge again, we simply increment the iterator of this leaf's edge set; if all level $i$ edges incident to this leaf are extracted, we update this leaf's bitmaps and move on.
%
% This description captures the hot path of the search, since in practice, in our experience the leaf queue is usually not empty.

When the search runs out of leaves in the leaf queue, it pops the next level $(i-1)$ cluster from the cluster queue, and uses the bitmaps to find the next unvisited leaf incident to this cluster, adding this leaf to the leaf queue.
If both queues become empty for a search, then the replacement search at this level is done (and we failed to find a replacement edge).
The effect of our optimization is to save a substantial amount of time performing downward traversals in the cluster forest to find nodes containing level $i$ edges, since we explicitly maintain a queue of them.

\subsection{Experiments}
Tables~\ref{tab:full_results1} and~\ref{tab:full_results2} present the raw un-normalized results of all of our experiments.

\begin{table*}
    \centering
    \caption{The first part of the raw results of our experiments. A red $\times$ indicates that the experiment timed out after 24 hours.}
    \label{tab:full_results1}
    \begin{tabular}{|l|l|l l l l l l l l|}
    \hline
        ~ & ~ & GER & USA & HH & CHEM & YT & POKE & WT & EW \\ \hline
        Total Insertion Time(s) & CF-LCA & 37.58 & 69.79 & 13.68 & 39.08 & 5.40 & 53.93 & 61.33 & 263.53 \\
        plus Initialization & CF-ROOT & 40.51 & 77.68 & 13.19 & 34.75 & 4.91 & 48.04 & 56.12 & 233.66 \\
        ~ & HDT & 78.69 & 150.49 & 20.08 & 53.00 & 10.19 & 67.50 & 77.05 & 354.72 \\
        ~ & DTree & 165.85 & 302.87 & 735.97 & {\color{red}$\times$} & 20.76 & 151.89 & 135.50 & 488.71 \\ \hline
        Total Deletion Time(s) & CF-LCA & 144.25 & 242.11 & 32.46 & 107.61 & 10.72 & 78.62 & 79.79 & 315.11 \\
        ~ & CF-ROOT & 224.74 & 440.04 & 53.48 & 167.14 & 15.20 & 106.27 & 108.50 & 384.05 \\
        ~ & HDT & 879.74 & 1792.04 & 171.29 & 534.76 & 34.22 & 208.75 & 196.19 & 733.93 \\
        ~ & DTree & 389.72 & 671.36 & 7234.59 & {\color{red}$\times$} & 61.78 & 561.40 & 456.83 & 1393.01 \\ \hline
        Total Update Time (s) & CF-LCA & 181.83 & 311.90 & 46.14 & 146.69 & 16.11 & 132.55 & 141.12 & 578.65 \\
        ~ & CF-ROOT & 265.25 & 517.72 & 66.67 & 201.88 & 20.11 & 154.31 & 164.62 & 617.71 \\
        ~ & HDT & 958.43 & 1942.53 & 191.37 & 587.76 & 44.41 & 276.25 & 273.24 & 1088.65 \\
        ~ & DTree & 555.56 & 974.23 & 7970.56 & {\color{red}$\times$} & 82.54 & 713.29 & 592.32 & 1881.71 \\ \hline
        Total Query Time (s) & CF-LCA & 16.65 & 15.88 & 21.93 & 31.03 & 13.25 & 21.40 & 20.19 & 25.79 \\
        ~ & CF-ROOT & 10.36 & 9.67 & 11.44 & 16.37 & 8.42 & 15.13 & 15.20 & 20.43 \\
        ~ & HDT & 17.12 & 16.99 & 19.52 & 27.37 & 14.04 & 25.46 & 25.74 & 38.35 \\
        ~ & DTree & 44.83 & 43.81 & 570.77 & {\color{red}$\times$} & 33.10 & 57.89 & 46.38 & 51.63 \\ \hline
        Peak Space Usage (GB) & CF-LCA & 4.85 & 9.24 & 0.98 & 2.06 & 0.45 & 1.07 & 1.12 & 3.37 \\
        ~ & CF-ROOT & 3.91 & 7.54 & 0.81 & 1.69 & 0.41 & 0.95 & 1.03 & 3.14 \\
        ~ & HDT & 68.99 & 140.59 & 11.45 & 25.75 & 6.20 & 8.51 & 9.11 & 27.96 \\
        ~ & DTree & 12.32 & 23.58 & 5.52 & 6.30 & 1.32 & 4.42 & 5.03 & 15.81 \\ \hline
    \end{tabular}
\end{table*}

\begin{table*}
    \centering
    \caption{The second part of the raw results of our experiments. A red $\times$ indicates that the experiment timed out after 24 hours.}
    \label{tab:full_results2}
    \begin{tabular}{|l|l|l l l l l l l l|}
    \hline
        ~ & ~ & SKIT & SO & LJ & ORK & TWIT & FR & GRID & RMAT \\ \hline
        Total Insertion Time(s) & CF-LCA & 24.90 & 66.48 & 115.70 & 337.34 & 7507.84 & 13136.54 & 19.69 & 3789.38 \\
        plus Initialization & CF-ROOT & 21.68 & 63.85 & 107.11 & 289.87 & 6311.46 & 12044.12 & 23.56 & 3208.91 \\
        ~ & HDT & 33.71 & 107.66 & 165.54 & 433.26 & 8925.71 & 17770.34 & 51.91 & 5221.15 \\
        ~ & DTree & 68.65 & 149.56 & 304.81 & 692.15 & 6977.36 & 12927.69 & 101.65 & 5159.71 \\ \hline
        Total Deletion Time(s) & CF-LCA & 36.70 & 76.52 & 168.64 & 353.04 & 9375.22 & 14557.69 & 53.67 & 5457.89 \\
        ~ & CF-ROOT & 51.14 & 105.88 & 241.87 & 422.26 & 9645.76 & 18242.89 & 122.22 & 8292.00 \\
        ~ & HDT & 110.12 & 224.97 & 566.28 & 732.94 & 14433.76 & 57307.62 & 325.34 & 20624.58 \\
        ~ & DTree & 231.77 & 436.07 & 1049.29 & 2196.70 & 16521.94 & 54575.80 & 208.84 & 16840.09 \\ \hline
        Total Update Time (s) & CF-LCA & 61.61 & 143.00 & 284.35 & 690.37 & 16883.06 & 27694.24 & 73.36 & 9247.27 \\
        ~ & CF-ROOT & 72.81 & 169.73 & 348.98 & 712.13 & 15957.22 & 30287.02 & 145.77 & 11500.91 \\
        ~ & HDT & 143.83 & 332.63 & 731.82 & 1166.21 & 23359.47 & 75077.96 & 377.25 & 25845.73 \\
        ~ & DTree & 300.42 & 585.63 & 1354.10 & 2888.86 & 23499.31 & 67503.50 & 310.49 & 21999.80 \\ \hline
        Total Query Time (s) & CF-LCA & 19.03 & 21.03 & 28.49 & 27.72 & 59.87 & 82.87 & 7.01 & 73.69 \\
        ~ & CF-ROOT & 12.91 & 15.64 & 19.87 & 18.81 & 47.21 & 74.98 & 4.07 & 57.06 \\
        ~ & HDT & 22.27 & 28.27 & 35.16 & 35.00 & 67.33 & 102.24 & 8.08 & 78.46 \\
        ~ & DTree & 44.00 & 40.02 & 59.72 & 60.71 & 59.92 & 83.14 & 37.87 & 71.51 \\ \hline
        Peak Space Usage (GB) & CF-LCA & 0.82 & 1.89 & 2.68 & 3.30 & 41.46 & 56.54 & 3.64 & 37.47 \\
        ~ & CF-ROOT & 0.75 & 1.81 & 2.37 & 3.02 & 39.80 & 52.88 & 3.01 & 33.62 \\
        ~ & HDT & 8.39 & 33.07 & 29.01 & 18.73 & 305.77 & 495.55 & 59.38 & 493.32 \\
        ~ & DTree & 2.92 & 5.54 & 9.72 & 17.72 & 184.26 & 274.39 & 10.79 & 140.56 \\ \hline
    \end{tabular}
\end{table*}

\end{appendix}

\end{document}
\endinput
%%
%% End of file `sample-sigconf.tex'.